\documentclass[12pt]{article}

\usepackage{latexsym}
\usepackage{amsmath}
\usepackage{amssymb}
\usepackage{amsthm}
\usepackage{amscd}
\usepackage[mathscr]{eucal}
\usepackage{fullpage}
\usepackage{graphicx}
\usepackage{subfigure}
\usepackage{psfrag}
\usepackage{rotating}
\usepackage{appendix}

\usepackage{color}

\newcommand{\R}{\mathbb{R}}

\newtheorem{theorem}{Theorem}
\newtheorem{corollary}{Corollary}

\newtheorem{lemma}{Lemma}

\newtheorem{definition}{Definition}

\usepackage{fullpage}
\numberwithin{equation}{section}

\title{\Huge{Multi-body spherically symmetric steady states of Newtonian self-gravitating elastic matter}}
 \author
 {A.~Alho  \\
           {\small Center for Mathematical Analysis, Geometry and Dynamical Systems}  \\
      {\small Instituto Superior T\'ecnico, Universidade de Lisboa} \\
      {\small Av. Rovisco Pais, 1049-001 Lisboa, Portugal} \\
       {\small  \tt   aalho@math.ist.utl.pt}\\[0.4cm]
       S. Calogero  \\
       {\small Department of Mathematical Sciences}  \\
       {\small Chalmers University of Technology, University of Gothenburg}  \\
       {\small Gothenburg, Sweden} \\
       {\small  \tt  calogero@chalmers.se}
       }
\date{}

\begin{document}
\maketitle
\begin{abstract}
We study the problem of static, spherically symmetric, self-gravitating elastic matter distributions in Newtonian gravity. 
To this purpose we first introduce a new definition of homogeneous, spherically symmetric (hyper)elastic body in Euler coordinates, i.e., in terms of matter fields defined on the current physical state of the body.  We show that our definition is equivalent to the classical one existing in the literature and which is given in Lagrangian coordinates, i.e., in terms of the deformation of the body from a given reference state.
After a number of well-known examples of constitutive functions of elastic bodies are re-defined in our new formulation, a detailed study of the Seth model is presented.  For this type of material the existence of single and multi-body solutions is established.
\end{abstract}
\section{Introduction}

In the Euler formulation of continuum mechanics, the configuration of static, self-gravitating matter distributions in Newtonian gravity is described by the equation
\begin{equation}\label{generaleq}
-\mathrm{Div}\,\sigma+\rho\nabla V=0,\quad V(x)=-\int_{\Omega}\frac{\rho(y)}{|x-y|}\,dy,\quad x\in\Omega,
\end{equation}
where $\sigma=\sigma(x)$ and $\rho=\rho(x)$ are the Cauchy stress tensor and the mass density of the matter, respectively, $V=V(x)$ is the gravitational potential self-induced by the matter distribution and $\Omega\subset\R^3$ is the interior of the matter support. We set $G=1$, where $G$ is Newton's gravitational constant. Equation~\eqref{generaleq} must be complemented by a constitutive equation relating the stress tensor and the mass density of the matter and which depends on the type of material considered. 
 

The constitutive equation $\sigma_{ij}=-f(\rho)\delta_{ij}$ defines a barotropic fluid with equation of state $p=f(\rho)$, where $p$ is the fluid pressure. In this case solutions of~\eqref{generaleq} are necessarily spherically symmetric~\cite{Li,Lin} and the pressure is a decreasing function of the distance from the center of symmetry, which entails that any static, self-gravitating fluid matter distribution with bounded support is {\it a priori} a single ball of matter. The existence of such fluid balls has been established for very general equation of state functions $f$, see~\cite{HU,RR,GR3}. When the matter distribution is made of collisionless kinetic particles (Vlasov matter), the existence of static balls is also well-understood, including their stability~\cite{HUR,LMR,GR3}, while the problem without symmetry restrictions is in large part still open; see~\cite{GR2,Sch} for results on axially symmetric solutions. The existence and stability of static self-gravitating shells of Vlasov matter has been proved in~\cite{GR,As}.

In this paper we consider matter distributions that consist of single or multiple elastic bodies, for which an explicit constitutive equation relating $\sigma$ to $\rho$ is not available in general. In fact elastic bodies are commonly defined using the Lagrangian formulation of continuum mechanics, that is to say, by specifying how $\sigma$ and $\rho$ depend on the deformation $\psi:\mathcal{B}\to\R^3$ of the body from a given reference configuration thereof in which the body occupies the domain $\mathcal{B}\subset\R^3$, see~\cite{ciarlet} and Appendix~\ref{ssg} below for details. One drawback of the Lagrangian formulation is that it generally provides only limited information on the region $\Omega=\psi(\mathcal{B})$ occupied by the matter distribution in its actual physical state and, consequently, on related properties, such as the formation of multiple bodies and the shape of the physical boundary $\partial\Omega$. 
In addition to this, within the applications in astrophysics, where bodies represent stars, planets, etc., the physical interpretation of the ``reference configuration" is questionable, as these bodies are only observable in their current deformed state. 

The above discussion suggests that a pure Euler formulation of the problem of self-gravitating elastic bodies in equilibrium may be desirable. In this paper we derive this formulation in the case of spherically symmetric configurations. More specifically we show in Appendix~\ref{ssg} that under the {\it a priori} assumption that the matter distribution is spherically symmetric, the constitutive equation of a homogeneous (hyper)elastic body can be written in the explicit form $\sigma=\sigma(\delta,\eta)$, where $\delta(r)$ and $\eta(r)$ are dimensionless quantities that measure respectively the density and the local mass of the body in units of the same quantities associated to a given matter distribution with constant density. In Section~\ref{sec2} we use this new form of the constitutive equation to {\it define} homogeneous spherically symmetric (hyper)elastic bodies in physical space. This definition 
is slightly different for elastic balls and elastic shells and is not yet completely independent of the reference state of the body, as it still involves the (constant) density of the material and the inner radius of the shell in the reference state. However we show that these two reference constants can be expressed in terms of physical space parameters. 
In Section~\ref{examples} we present the constitutive equation $\sigma=\sigma(\delta,\eta)$ for some well-known examples of elastic materials which have important applications in physics, chemistry and other sciences. 
In Section~\ref{sethsec} we focus the attention on the Seth elasticity model~\cite{seth}. For this particular elastic matter model we prove the existence of single self-gravitating balls, as well as  the existence of multi-body distributions consisting of an interior ball, or vacuum region, surrounded by an arbitrary number of shells. We choose the Seth model on the one hand because it leads to relatively simple equations and, on the other hand, because we can construct an explicit self-similar solution of the equations for the density $\rho$, which makes it possible to analyze the large radius behavior of general solutions by using the powerful methods of dynamical systems theory. In a sequel of the present paper we will extend our results on self-gravitating elastic balls to more general constitutive equations, including all the examples listed in Section~\ref{examples} below.

We conclude this Introduction by mentioning some works on static, self-gravitating elastic bodies in Newtonian gravity which are related to our results; see~\cite{CH} for a short review on the analogous problem in General Relativity. A survey of applications to planetary objects, including a detailed historical account to the problem, can be found in the recent monograph~\cite{MW16}. A very influential reference in the early stage of the theory is the treatise by Love~\cite{Lov27}. More recently, the existence of single body configurations near the reference state (small deformations) is proved in~\cite{BS} using the implicit function theorem in Lagrangian coordinates. The special case of the latter result for spherically symmetric static shells, in Newtonian and Einsteinian gravity, is discussed in~\cite{LVK}.
For large deformations, the distinction between the reference and the deformed state of the body becomes important and the problem considerably more difficult. In~\cite{tomsimo} this problem was studied using variational calculus methods, but the only property which was possible to prove on the regularity of the physical boundary $\partial\Omega$ of the body was that   $\partial\Omega$ has zero Lebesgue measure.
   To the authors knowledge, the results presented in this paper are the first analytical results on the existence of static, self-gravitating bodies in the non-linear theory elasticity with no restriction on the amount of internal strain of the matter.
\section{Spherically symmetric steady states of self-gravitating elastic matter}\label{sec2}

In spherical symmetry we have {\it a priori} that $\rho=\rho(r)$, where $r=|x|$, and the Cauchy stress tensor has the form 
\[
\sigma_{ij}(x)=-p_\mathrm{rad}(r)\frac{x_ix_j}{r^2}-p_\mathrm{tan}(r)\left(\delta_{ij}-\frac{x_ix_j}{r^2}\right),
\]
where $p_\mathrm{rad}(r)$ and $p_\mathrm{tan}(r)$ are called respectively radial and tangential stress (or pressure). 
Moreover 
\[
\nabla V(r)=\frac{m(r)}{r^2}\frac{x}{r},\quad \text{where}\quad m(r)=4\pi\int_0^r\rho(s)s^2ds
\]
is the mass of the matter distribution enclosed in the ball of radius $r>0$.
We abuse slightly the notation by using the same symbol to denote spherically symmetric functions in Cartesian and spherical coordinates (e.g., $\rho(x)=\rho(r)$). For spherically symmetric configurations~\eqref{generaleq} reduces to
\begin{equation}\label{generalss}
p'_\mathrm{rad}(r)+2\,\frac{p_\mathrm{rad}(r)-p_\mathrm{tan}(r)}{r}+\rho\,\frac{m}{r^2}=0,\quad r\in\Omega,
\end{equation}
where
\[
\Omega=\mathrm{Int}\{r>0:\rho(r)>0\}
\]
is the interior of the matter support. We assume that $p_\mathrm{rad}>0$ and $p_\mathrm{tan}>0$ in $\Omega$; in particular we exclude from our analysis those hypothetical astrophysical objects, e.g., dark-matter stars, which are theorized to have negative interior pressures. Moreover we employ the standard boundary conditions for astrophysical systems surrounded by vacuum that the radial pressure should vanish on the boundary of the matter support~\cite{MW16}; see~\cite{KWW} for other boundary conditions used in astrophysics.

\begin{definition}\label{nbodydef}
A triple $(\rho,p_\mathrm{rad},p_\mathrm{tan})$ is said to be a spherically symmetric, static, self-gravitating $n$-body matter distribution with non-vacuum core if there exist
\[
0= r_0<r_1<\dots< r_{2n-1}
\]
such that
\begin{itemize}
\item[(i)]$\Omega=\cup_{j=1}^{n}I_j$, where $I_j:=(r_{2j-2},r_{2j-1})$,
\item[(ii)] $(\rho,p_\mathrm{rad},p_\mathrm{tan})\in C^1(\Omega)\cap C^0(\overline{\Omega})$ satisfy~\eqref{generalss} in $\Omega$, 
\item[(iii)] $p_\mathrm{rad},p_\mathrm{tan}$ are positive in $\Omega$, 
\item[(iv)] $p_\mathrm{rad}(r_j)=0$, for $j=1,\dots, 2n-1$, and $p_\mathrm{rad}(r_0)>0$,
\item[(v)] $p_\mathrm{rad}(0)=p_\mathrm{tan}(0)$,
\item[(vi)] $\rho(r)=p_\mathrm{tan}(r)=p_\mathrm{tan}(r)=0$, for $r\in [0,\infty)\setminus\overline{\Omega}$.
\end{itemize}
If the conditions $r_0=0$, $p_\mathrm{rad}(r_0)>0$ are replaced with $r_0>0$  and $p_\mathrm{rad}(r_0)=0$, then $(\rho,p_\mathrm{rad},p_\mathrm{tan})$ is said to be a spherically symmetric, static, self-gravitating $n$-body matter distribution with vacuum core. 
The restriction of $(\rho,p_\mathrm{rad},p_\mathrm{tan})$ into the closed interval $\overline{I_j}$, $j=1,\dots, n$, will be called the $j^\mathrm{th}$ body in the matter distribution and denoted by $\mathfrak{B}_j$.
The local mass of $\mathfrak{B}_j$ is
\[
m_j(r)=4\pi\int_{r_{2j-2}}^r\rho(s)\,s^2ds,\quad r\in \overline{I_j},
\] 
and its total mass is
\[
M_j=m(r_{2j-1})-m(r_{2j-2}),\quad j=1,\dots, n.
\]
The total mass of the $n$-body distribution is $M=M_1+\dots +M_n=m(r_{2n-1})$.
\end{definition}
Note that each single body $\mathfrak{B}_1,\dots, \mathfrak{B}_n$ in a multi-body distribution may be made of a different type of material. If $2\leq j\leq n$, then each $\mathfrak{B}_j$ is necessarily a shell of matter, while $\mathfrak{B}_1$ is either a shell, if $r_0>0$, or a ball, if $r_0=0$. By definition we have
\[
(\rho,p_\mathrm{rad},p_\mathrm{tan})=\mathfrak{B}_1+\dots+\mathfrak{B}_n,\quad r\in\overline{\Omega}.
\]

{\it Remark.} The number of bodies in a static, self-gravitating matter distribution may be restricted by the type of material(s) considered. For instance, for perfect fluids the pressure is a decreasing function of the radius and thus static, self-gravitating fluids can only exist in a single ball configuration. When the type of material(s) allows for the existence of multiple bodies, one can construct $n$-body matter distributions by gluing together single body solutions (e.g., a ball with single disjoint shells). 

{\it Remark.} Condition (v) ensures that the Cauchy stress tensor is continuous at the center.

In the remainder of this section we discuss what it means that a body is made of elastic matter. As mentioned in the Introduction, the standard definition of elastic body is given in the Lagrangian formulation of continuum mechanics, i.e., by specifying how $\sigma$ and $\rho$ depend on the deformation of the body from a given reference state, see Appendix~\ref{ssg}. Our purpose is to introduce an alternative definition in Euler coordinates, i.e., in terms of matter fields defined on the actual physical state of the body. We need to distinguish the cases when the body is a ball and a shell, and thus introduce the following two definitions:  
\begin{definition}\label{elasticdef}
Let $\mathfrak{B}=(\rho,p_\mathrm{rad},p_\mathrm{tan})$ be a static, self-gravitating ball supported in $\overline{I}$, where $I\subset(0,\infty)$ is an open, bounded interval such that $0\in\overline{I}$; let $m(r)$ be the local mass of the ball.
Given a constant $\mathcal{K}>0$ and two functions $\widehat{p}_\mathrm{rad},\widehat{p}_\mathrm{tan}:(0,\infty)^2\to\R$ independent of $\mathcal{K}$, we say that $\mathfrak{B}$ is made of homogeneous elastic matter with constitutive functions $\widehat{p}_\mathrm{rad},\widehat{p}_\mathrm{tan}$ and that $\mathfrak{B}$ has reference density $\mathcal{K}>0$ if the radial and tangential pressures have the form
\[
p_\mathrm{rad}(r)=\widehat{p}_\mathrm{rad}(\delta(r),\eta(r)),\quad p_\mathrm{tan}(r)=\widehat{p}_\mathrm{tan}(\delta(r),\eta(r)),
\]
where 
\[
\delta(r)=\frac{\rho(r)}{\mathcal{K}},\quad
\eta(r)=\frac{m(r)}{\tfrac{4\pi}{3} \mathcal{K} r^3},
 \quad r\in I.
\]
\end{definition}
\begin{definition}\label{elasticdefshell}
Let $\mathfrak{B}=(\rho,p_\mathrm{rad},p_\mathrm{tan})$ be a static, self-gravitating shell supported in $\overline{I}$, where $I\subset(0,\infty)$ is an open, bounded interval such that $0\notin\overline{I}$; let $m(r)$ be the local mass of the shell.
Given two constants $\mathcal{K},\mathcal{S}>0$, and two functions $\widehat{p}_\mathrm{rad},\widehat{p}_\mathrm{tan}:(0,\infty)^2\to\R$ independent of $\mathcal{K},\mathcal{S}$, we say that $\mathfrak{B}$ is made of homogeneous elastic matter with constitutive functions $\widehat{p}_\mathrm{rad},\widehat{p}_\mathrm{tan}$ and that $\mathfrak{B}$ has reference density $\mathcal{K}>0$ and reference inner radius $\mathcal{S}$ if the radial and tangential pressures have the form
\[
p_\mathrm{rad}(r)=\widehat{p}_\mathrm{rad}(\delta(r),\eta(r)),\quad p_\mathrm{tan}(r)=\widehat{p}_\mathrm{tan}(\delta(r),\eta(r)),
\]
\[
\delta(r)=\frac{\rho(r)}{\mathcal{K}},\quad
\eta(r)=\left(\frac{\mathcal{S}}{r}\right)^3+\frac{m(r)}{\tfrac{4\pi}{3} \mathcal{K} r^3},
 \quad r\in I.
\]
\end{definition}
A particular important case is when the elastic material making up the body is hyperelastic.
\begin{definition}\label{hyperelasticdef}
Let $\mathcal{B}$ be a static, self-gravitating ball or shell of homogeneous elastic matter.
If there exists a function $w:(0,\infty)^2\to\R$ such that $w(1,1)=0$ and
\begin{equation}\label{hyperdef}
\widehat{p}_\mathrm{rad}(\delta,\eta)=\delta^2\partial_\delta w(\delta,\eta),\quad \widehat{p}_\mathrm{tan}(\delta,\eta)=\widehat{p}_\mathrm{rad}(\delta,\eta)+\frac{3}{2}\delta\eta\partial_\eta w(\delta,\eta),
\end{equation}
then the body is said to be made of hyperelastic matter with stored energy function $w$.
\end{definition}
A body made of (hyper)elastic matter will also be called (hyper)elastic body for short.
In Appendix~\ref{ssg} we show that Definitions~\ref{elasticdef}-\ref{hyperelasticdef} are formally equivalent to the standard definitions, given in Lagrangian coordinates, of homogeneous, (hyper)elastic balls/shells, when the deformation and the Piola-Kirchhoff stress tensor are assumed to be spherically symmetric. We show in Appendix~\ref{varsec} that the equation~\eqref{generalss} for hyperelastic bodies admits a variational formulation. 

{\it Remark.} Hyperelastic materials are the only admissible elastic matter models in General Relativity because the stored energy function is a source term for the gravitational field in Einstein's theory.

{\it Remark.} When the stored energy function does not depend on $\eta$, i.e., $w(\delta,\eta)=w(\delta)$, the elastic material reduces to barotropic fluid matter with pressure $p=\delta^2w'(\delta)=f(\rho)$.

{\it Remark.} As suggested by the terminology, $\mathcal{K}$ corresponds to the density of the body in the reference state. Similarly, the reference inner radius $\mathcal{S}$ of the shell corresponds to the inner radius of the shell in the reference state, see Appendix~\ref{ssg}. These two constants are the only link with the reference state in the definitions above. An alternative interpretation of the constants $\mathcal{K},\mathcal{S}$ in physical space is given at the end of this section.

{\it Remark.} We emphasize that the reference parameters $\mathcal{K},\mathcal{S}$ are {\it not} material constants, that is to say, balls and shells of a given elastic material, i.e., with given constitutive functions $\widehat{p}_\mathrm{rad}$, $\widehat{p}_\mathrm{tan}$,  can have different values of the reference parameters $\mathcal{K}$, $\mathcal{S}$.   
%

{\it Remark.} Definitions~\ref{elasticdef}--\ref{hyperelasticdef} could be extended to include inhomogeneous materials by letting the reference density $\mathcal{K}$ depend on $r>0$, but for simplicity we shall not do so. The class of homogeneous materials already encompasses the most common examples of elastic matter found in the literature, see Section~\ref{examples}, and includes in particular the case of barotropic fluid matter, which is one of the most popular matter models in astrophysics, see~\cite{KWW}.

%

We now translate in our formulation a number of standard assumptions on the constitutive functions of elastic matter. The usual formulation of these assumptions in Lagrangian coordinates is given in Appendix~\ref{ssg}.
The first assumption is that the reference state of the body should be a natural state, i.e., stress-free. Since in our formulation the reference state of the body corresponds to a state in which it has constant density $\mathcal{K}$, and, in the case of a shell, inner radius $\mathcal{S}$, then $\delta=\eta=1$ holds for a body in its reference state. Hence we require the condition
\begin{subequations}\label{lincomeul}
\begin{equation}\label{naturalstatecon}
\widehat{p}_\mathrm{rad}(1,1)=\widehat{p}_\mathrm{tan}(1,1)=0.
\end{equation}
The second assumption on the constitutive functions demands that for infinitesimal deformations the body should behave according to Hooke's law of linear elasticity. As shown in Appendix~\ref{ssg}, this condition is achieved by postulating the existence of constants $\lambda,\mu$, called Lam\'e coefficients, such that  
\begin{align}
&\partial_\delta\widehat{p}_\mathrm{rad}(1,1)=\lambda+2\mu,\quad \partial_\eta\widehat{p}_\mathrm{rad}(1,1)=-\frac{4}{3}\mu,\label{com1}\\
&\partial_\delta\widehat{p}_\mathrm{tan}(1,1)=\lambda,\quad\partial_\eta\widehat{p}_\mathrm{tan}(1,1)=\frac{2\mu}{3}\label{com2}.
\end{align}

{\it Remark.} The Poisson ratio of an elastic material is defined in terms of the Lam\'e coefficients as
$\nu=(\lambda+\mu)^{-1}\lambda/2$. 
The Poisson ratio of most materials (e.g., metals, rubber, etc.) lies in the interval $\nu\in (0,1/2)$. Materials with $\nu<0$, e.g., paper, are called auxetic.

To justify our next assumption, let $\rho_c=\rho(0)$ be the central density of an elastic ball. As   
\[
\lim_{r\to 0^+}\eta(r)=\frac{\rho_c}{\mathcal{K}}=\delta(0),
\]
the central pressures of the ball are
\[
p_\mathrm{rad}(0)=\widehat{p}_\mathrm{rad}(\delta(0),\delta(0)),\quad p_\mathrm{tan}(0)=\widehat{p}_\mathrm{tan}(\delta(0),\delta(0)).
\] 
Thus, in order to ensure that condition (v) in Definition~\ref{nbodydef} is satisfied for all possible values of the central density, we assume
\begin{equation}
\widehat{p}_\mathrm{rad}(\delta,\delta)=\widehat{p}_\mathrm{tan}(\delta,\delta),\quad \text{for all $\delta>0$.}
\end{equation}
\end{subequations}
%
In terms of the variables $(\delta,\eta,m)$ the system~\eqref{generalss} within each body reads
\begin{subequations}\label{generalss2}
\begin{align}
&\partial_\delta\widehat{p}_\mathrm{rad}(\delta,\eta)\delta'=-\frac{3}{r}\partial_\eta\widehat{p}_\mathrm{rad}(\delta,\eta)\left(\delta-\eta\right)-\frac{2}{r}(\widehat{p}_\mathrm{rad}(\delta,\eta)-\widehat{p}_\mathrm{tan}(\delta,\eta))-\mathcal{K}\,\delta\,\frac{m}{r^2},\label{eqdelta}\\
&\eta'=\frac{3}{r}\left(\delta-\eta\right),\\
&m'=4\pi\mathcal{K} r^2\delta,
\end{align}
\end{subequations}
where we recall that each body may have a different constitutive function as well as different reference parameters.

The system~\eqref{generalss2} must be complemented by appropriate initial conditions, which depend on the type and location of the body. In the case of a single ball, initial data are given at $r=0$ according to
\[
\delta(0)=\eta(0)=\delta_c:=\rho_c/\mathcal{K}>0,\quad m(0)=0,
\] 
where $\rho_c$ is the central density of the ball. If the body is a single shell with inner radius $r_0>0$, then the initial data for~\eqref{generalss2}  must be given at $r=r_0$ according to 
\[
\eta(r_0)=\left(\frac{\mathcal{S}}{r_0}\right)^3,\quad \delta(r_0):p_\mathrm{rad}(\delta(r_0),\eta(r_0))=0,\quad m(r_0)=0.
\]
If the body is a shell surrounding a $n$-body matter distribution with total mass $M=M_1+\dots+M_n$, then the initial data for~\eqref{generalss2} must be given at some radius $r_{2n}>r_{2n-1}$ according to
\[
\eta(r_{2n})=\left(\frac{\mathcal{S}}{r_{2n}}\right)^3,\quad \delta(r_{2n}):p_\mathrm{rad}(\delta(r_{2n}),\eta(r_{2n}))=0,\quad m(r_{2n})=M.
\]
\begin{definition}\label{regsol}
Given $0\leq r_0<R\leq \infty$, a regular solution  of~\eqref{generalss2} in the interval $(r_0,R)$ is a triple of non-negative functions $(\delta,\eta,m)\in C^1((r_0,R))\cap C^0([r_0,R))$ that satisfy~\eqref{generalss2} for $r\in (r_0,R)$. Moreover in the case $r_0=0$ we require
\begin{equation}\label{regcenter}
\lim_{r\to 0^+}\delta(r)=\lim_{r\to 0^+}\eta(r)=\delta_c,\quad\lim_{r\to 0^+}m(r)=0,
\end{equation} 
for some positive constant $\delta_c$.  Equation~\eqref{regcenter} will be referred to as the regular center condition.
\end{definition}
To conclude this section we want to discuss briefly the interpretation in physical space of the reference parameters $\mathcal{K},\mathcal{S}$. Consider first the case of an elastic ball with reference density $\mathcal{K}$ and total mass $M$. Let
\[
F(\delta)=\widehat{p}_\mathrm{rad}(\delta,\delta)=\widehat{p}_\mathrm{tan}(\delta,\delta).
\]
The central pressure $p_c$ of the ball is $p_c=F(\rho_c/\mathcal{K})$, where $\rho_c=\rho(0)$ is the central density.
Hence, provided $F$ is invertible, we obtain 
\begin{equation}\label{K}
\mathcal{K}=\frac{\rho_c}{F^{-1}(p_c)},
\end{equation}
which defines $\mathcal{K}$ in terms of the physical parameters $\rho_c,p_c$. Alternatively one may use the fact that at the radius $r_1$ of the ball there holds
\begin{equation}\label{K2}
\widehat{p}_\mathrm{rad}\left(\frac{\rho(r_1)}{\mathcal{K}},\frac{M}{\frac{4\pi}{3}\mathcal{K}r_1^3}\right)=0,
\end{equation}
which can be used to define $\mathcal{K}$ implicitly in terms of the physical parameters $\rho(r_1),r_1,M$. 
A similar interpretation can be given for the reference parameters of an elastic shell. Let $r_0$ be the inner radius of the shell in physical space. As $\eta(r_0)=(\mathcal{S}/r_0)^3$ and $p_\mathrm{rad}(r_0)=0$, then 
\begin{equation}\label{S1}
\widehat{p}_\mathrm{rad}\left(\frac{\rho(r_0)}{\mathcal{K}},\left(\frac{\mathcal{S}}{r_0}\right)^3\right)=0.
\end{equation}
Similarly, on the outer radius $r_1>0$ of the shell there holds
\begin{equation}\label{S2}
\widehat{p}_\mathrm{rad}\left(\frac{\rho(r_1)}{\mathcal{K}},\left(\frac{\mathcal{S}}{r_1}\right)^3+\frac{M}{\frac{4\pi}{3}\mathcal{K}r_1^3}\right)=0,
\end{equation}
where $M$ is the mass of the shell.
The system~\eqref{S1}-\eqref{S2} defines implicitly the reference parameters $\mathcal{K},\mathcal{S}$ in terms of the physical parameters $r_0,r_1,\rho(r_0),\rho(r_1)$ and $M$. Of course, explicit formula for the reference parameters in terms of the physical ones can be derived only for very simple constitutive functions, while in general one has to solve~\eqref{K}--\eqref{S2} numerically.
\section{Example of constitutive functions}\label{examples}
In this section we present a selection of constitutive functions for elastic materials using the formulation given in the previous section. We emphasize that the formula in this section are only valid for spherically symmetric distributions. For the analogous constitutive functions expressed in Lagrangian coordinates and with no symmetry restriction, see Appendix~\ref{ssg} and the references~\cite{BH,Drozdov, Lurie, TW}. All examples in this section verify the  conditions~\eqref{lincomeul}.
\subsubsection*{Seth materials}
The constitutive functions for spherically symmetric Seth materials are 
\begin{subequations}\label{sethmodel}
\begin{align}
&\widehat{p}_\mathrm{rad}(\delta,\eta)=\lambda\,\eta^{2/3}+\frac{\lambda+2\mu}{2}\eta^{-4/3}\delta^2-p_0,\\
&\widehat{p}_\mathrm{tan}(\delta,\eta)=(\lambda+\mu)\,\eta^{2/3}+\frac{\lambda}{2}\eta^{-4/3}\delta^2-p_0,\\
&p_0=\frac{3\lambda+2\mu}{2},
\end{align}
\end{subequations}
where the Lam\'e coefficients satisfy $\mu>0$ and $3\lambda+2\mu>0$ (i.e., $p_0>0$) and so $-1<\nu<1/2$. 
A detailed analysis of this model is presented in Section~\ref{sethsec}. The Seth model is not hyperelastic and therefore it is not suitable as elastic matter model in General Relativity.



\subsubsection*{Saint Venant-Kirchhoff materials}
The Saint Venant-Kirchhoff model is hyperelastic with stored energy function given by 
\begin{align*}
w(\delta,\eta)=\frac{1}{8} \left(\frac{\eta ^{4/3}}{\delta ^2}+\frac{2}{\eta ^{2/3}}-3\right)^2 (\lambda
   +2 \mu )+\mu  \left(\frac{\eta ^{4/3}}{\delta ^2}+\frac{2}{\eta
   ^{2/3}}-3\right)-\frac{\mu}{2}  \left(\frac{2 \eta ^{2/3}}{\delta ^2}+\frac{1}{\eta
   ^{4/3}}-3\right),
\end{align*}
where the Lam\'e coefficients satisfy $\mu>0$, $3\lambda+2\mu>0$.
It follows that
\begin{align*}
\widehat{p}_\mathrm{rad}(\delta,\eta) &=\mu\frac{\eta^{4/3}}{\delta}\left(1-\frac{\eta^{4/3}}{\delta^2}\right)+\lambda\frac{\eta^{2/3}}{2\delta}\left(3\eta^{2/3}-\frac{\eta^2}{\delta^2}-2\right), \\
\widehat{p}_\mathrm{tan}(\delta,\eta)& = \mu\frac{\delta}{\eta^{2/3}}\left(1-\eta^{-2/3}\right)+\lambda\frac{\delta}{2\eta^{2/3}}\left(3-2\eta^{-2/3}-\frac{\eta^{4/3}}{\delta^2}\right).
\end{align*}
\subsubsection*{Quasi-linear Signorini materials}
The (quasi-linear) Signorini model is hyperelastic with stored energy function given by
\begin{align*}
w(\delta,\eta) =\frac{1}{\delta}\left[\frac{1}{8} \left(\frac{\delta ^2}{\eta ^{4/3}}+2 \eta ^{2/3}-3\right)^2 (\lambda
   +\mu )+\frac{\mu}{2} \left(\frac{\delta ^2}{\eta ^{4/3}}+2 \eta
   ^{2/3}-1\right)\right]-\mu,
\end{align*}
where $\mu>0$ and $9\lambda+5\mu >0$, i.e., $-5/8<\nu<1/2$.
The constitutive functions for the radial and tangential pressure read
\begin{align*}
&\widehat{p}_\mathrm{rad}(\delta,\eta) =\frac{\lambda+\mu}{8}\eta^{4/3}\left(3\left(\frac{\delta}{\eta}\right)^4+4\left(\frac{\delta}{\eta}\right)^2-4\right)+\frac{3\lambda+\mu}{4}\eta^{2/3}\left(2-\left(\frac{\delta}{\eta}\right)^2\right)-p_0,\\
&\widehat{p}_\mathrm{tan}(\delta,\eta) =\frac{\lambda+\mu}{8}\eta^{4/3}\left(4-\left(\frac{\delta}{\eta}\right)^4\right)+\frac{3\lambda+\mu}{4}\frac{\delta^2}{\eta^{4/3}}-p_0,\\
&p_0=\frac{9\lambda+5\mu}{8}.
\end{align*}
\subsubsection*{Hadamard materials}
The Hadamard model is hyperelastic with stored energy function given by
\[
w(\delta,\eta)=\frac{1}{2} \left(\alpha \left(\frac{\eta ^{4/3}}{\delta ^2}+\frac{2}{\eta
   ^{2/3}}-3\right)+\beta \left(\frac{2 \eta ^{2/3}}{\delta ^2}+\frac{1}{\eta
   ^{4/3}}-3\right)+h\left(\delta^{-2}\right)-h(1)\right),
\]
where the compatibility conditions~\eqref{lincomeul} imply
\[
h'(1)=-(\alpha+2\beta),\quad h''(1)=\frac{\lambda+2\mu}{2},\quad \alpha+\beta=\mu.
\] 
The constitutive functions for the radial and tangential pressure are
\begin{align*}
&\widehat{p}_\mathrm{rad}(\delta,\eta)=-\frac{1}{\delta}  \left(\alpha \eta ^{4/3}+2 \beta \eta ^{2/3}+h'\left(\delta^{-2}\right)\right),\\ 
&\widehat{p}_\mathrm{tan}(\delta,\eta)=-\frac{1}{\delta}\left[\alpha\eta^{4/3}\left(\frac{\delta}{\eta}\right)^2+\beta\eta^{2/3}\left(1+\left(\frac{\delta}{\eta}\right)^2\right)+h'\left(\delta^{-2}\right)\right].
\end{align*}
\subsubsection*{Materials with linear constitutive function}
Due to the compatibility conditions~\eqref{lincomeul}, the constitutive functions of spherically symmetric elastic bodies have all the same linear approximation, namely 
\begin{align*}
&\widehat{p}_\mathrm{rad}(\delta,\eta)=(\lambda+2\mu)\delta-\frac{4\mu}{3}\eta-p_0,\\
&\widehat{p}_\mathrm{tan}(\delta,\eta)=\lambda\,\delta+\frac{2\mu}{3}\eta-p_0,\\
&p_0=\frac{3\lambda+2\mu}{3}.
\end{align*}
Materials with linear constitutive function are hyperelastic with stored energy function given by
\[
w(\delta,\eta)=(\lambda+2\mu)\log\delta-\frac{4\mu}{3}\log\eta+\frac{4\mu}{3}\frac{\eta}{\delta}+\frac{3\lambda+2\mu}{3\delta}-\lambda-2\mu.
\]
\section{Analysis of the Seth model}\label{sethsec}
In this section we present a detailed analysis of the Seth model, for which the constitutive functions of the principal pressures are given in terms of the Lam\'e coefficients by~\eqref{sethmodel}. Although this is not always necessary in what follows, we assume that
\[
\lambda>0,\quad\mu>0,\quad \text{hence}\quad  \nu=\frac{\lambda}{2(\lambda+\mu)}\in (0,1/2),
\]
which are conditions satisfied by most materials. As
$\partial_\delta\widehat{p}_\mathrm{rad}(\delta,\eta)>0$ holds in the matter interior,
the system~\eqref{generalss2} for Seth materials becomes
\begin{subequations}\label{Setheqs}
\begin{align}
&\delta'=-\frac{3}{r}\theta_1(\delta,\eta)(\delta-\eta)-\frac{2}{r}\theta_2(\delta,\eta)-\theta_3(\delta,\eta)\frac{m}{r^2}\delta,\\
&\eta'=\frac{3}{r}(\delta-\eta),\quad m'=4\pi\mathcal{K} r^2\delta,
\end{align}
where
\begin{align}
&\theta_1(\delta,\eta):=\frac{\partial_\eta\widehat{p}_\mathrm{rad}(\delta,\eta)}{\partial_\delta\widehat{p}_\mathrm{rad}(\delta,\eta)}=\frac{2}{3(\lambda+2\mu)}\left(\lambda\frac{\eta}{\delta}-(\lambda+2\mu)\frac{\delta}{\eta}\right),\\[5mm]
&\theta_2(\delta,\eta):=\frac{\widehat{p}_\mathrm{rad}(\delta,\eta)-\widehat{p}_\mathrm{tan}(\delta,\eta)}{\partial_\delta\widehat{p}_\mathrm{rad}(\delta,\eta)}=\frac{\mu}{\lambda+2\mu}\frac{\delta^2-\eta^2}{\delta},\\[5mm]
&\theta_3(\delta,\eta):=\frac{\mathcal{K}}{\partial_\delta\widehat{p}_\mathrm{rad}(\delta,\eta)}=\frac{\mathcal{K}}{\lambda+2\mu}\frac{\eta^{4/3}}{\delta}.
\end{align}
\end{subequations}

The equations~\eqref{Setheqs} admit the (non-regular) self-similar solution $(\delta^\star,\eta^\star,m^\star)\in C^1((0,\infty))$ given by
\begin{equation}\label{sethexact}
\delta^\star(r)=\frac{c}{\mathcal{K}}\frac{1}{r^{3/2}},\quad \eta^\star(r)=\frac{2c}{\mathcal{K}}\frac{1}{r^{3/2}},\quad m^\star(r)=\frac{8\pi c}{3}r^{3/2},
\end{equation}
where
\[
c=\frac{\left(\frac{3}{\pi }\right)^{3/4} (9\lambda +14 \mu )^{3/4}}{16 \sqrt{\mathcal{K} }}.
\]
The associated principal pressures are given by
\begin{align*}
&p_\mathrm{rad}^\star(r)=-p_0+\frac{9\lambda+2\mu}{8}\left(\frac{2c}{\mathcal{K}}\right)^{2/3}\frac{1}{r},\\
& p_\mathrm{tan}^\star(r)=-p_0+\frac{9\lambda+8\mu}{8}\left(\frac{2c}{\mathcal{K}}\right)^{2/3}\frac{1}{r}.
\end{align*}
Note that $p^\star_\mathrm{tan}(r)>p^\star_\mathrm{rad}(r)$, for all $r>0$, the radial pressure is positive for $r<R$, negative for $r>R$ and vanishes at $r=R$, where
\[
R=\left(\frac{2c}{\mathcal{K}}\right)^{2/3}\frac{9\lambda+8\mu}{12\lambda+8\mu}.
\]
Hence the self-similar solution truncated at $r=R$ describes a static, self-gravitating ball with irregular center.
 It is convenient introduce the new variables  
\[
x(r)=\vartheta\,\eta(r)^{2/3},\quad y(r)=\frac{\delta(r)}{\eta(r)},\quad z=\frac{3}{4\pi\mathcal{K}}\frac{m(r)}{r^3\eta(r)},
\]
where
\[
\vartheta=\mathcal{K}\,\sqrt{\frac{4\pi}{3(\lambda+2\mu)}}.
\]
Note that $z=1$ holds for balls, while $z<1$ for shells.
In terms of the new variables the system~\eqref{Setheqs} for regular solutions becomes
\begin{subequations}\label{generalssSeth}
\begin{align}
&x'=-\frac{2x}{r}(1-y),\label{xeq}\\
&y'=\frac{1}{r}(a+b y+y^2)\frac{1-y}{y}-rx^2z,\\
&z'=\frac{3y}{r}(1-z),
\end{align}
\end{subequations}
where
\[
a=\frac{2(\lambda+\mu)}{\lambda+2\mu}\geq 1,\quad b=\frac{2\mu}{\lambda+2\mu}\leq 1
\]
are positive material constants.
The self-similar solution takes the form
\[
x^\star(r)=\left(\frac{2c}{\mathcal{K}}\right)^{2/3}\frac{\vartheta}{r},\quad y^\star(r)=\frac{1}{2},\quad z^\star(r)=1.
\]
We call a solution of~\eqref{generalssSeth} regular if the corresponding solution of~\eqref{Setheqs} is regular in the sense of Definition~\ref{regsol}. In particular the regular center condition~\eqref{regcenter} in the new variables becomes
\[
\lim_{r\to 0^+}y(r)=\lim_{r\to 0^+}z(r)=1.
\]
\begin{theorem}\label{cauchyprob}
Let $r_0\geq0$, $x_0>0$, $y_0>0, z_0\geq 0$ be given with $y_0=z_0=1$ if $r_0=0$. There exists a unique regular solution $(x(r),y(r),z(r))\in C^1([r_0,\infty))$ of~\eqref{generalssSeth} satisfying $\lim_{r\to r_0^+}(x(r),y(r),z(r))=(x_0,y_0,z_0)$. Moreover the following holds:
\begin{itemize}
\item[(i)] if $y_0\leq1$, then $y(r)<1$ and $x(r)< x_0$, for all $r>r_0$;
\item[(ii)] $(x(r),y(r),z(r))\sim (x^\star(r),y^\star(r),z^\star(r))$ as $r\to \infty$.
\end{itemize}
\end{theorem}
\begin{proof}

Local existence and uniqueness of regular solutions in a right interval $[r_0,r_0+\varepsilon)$ for some $\varepsilon>0$ follows by standard ODE theory when $r_0>0$, while for $r_0=0$ it follows by results proved e.g. in~\cite{RS} (see also~\cite{P} for a detailed proof of a similar result in the general relativistic case). Moreover the solution can be continued uniquely up to $R>r_0+\varepsilon$ as long as $\inf_{(r_0,R)}y(r)>0$, $\sup_{(r_0,R)}y(r)<\infty$, $\sup_{(r_0,R)}x(r)<\infty$ and $\sup_{(r_0,R)}z(r)<\infty$. As shown by the dynamical systems analysis below, these bounds are verified {\it a priori} for all $R>r_0$,  
hence the proof of the first statement of the theorem will follow {\it a posteriori} by this analysis.
 
Let $r>r_0$, $u(r)=rx(r)$, and introduce the new independent variable $\xi$ by
\[
\frac{d}{d\xi}=ry\frac{d}{dr}.
\]
The system~\eqref{generalssSeth} in the new variables becomes the autonomous dynamical system
\begin{equation}\label{newdynsys}
\frac{du}{d\xi}=-u(1-2y)y,\quad\frac{dy}{d\xi}=(a+by+y^2)(1-y)-u^2yz,\quad \frac{dz}{d\xi}=3y^2(1-z).
\end{equation}
Note that $u=0$ and $z=1$ are invariant surfaces for the dynamical system~\eqref{newdynsys}. Moreover, since 
\[
\left(\frac{dy}{d\xi}\right)_{y=0}=a>0,\quad \left(\frac{dz}{d\xi}\right)_{z=0}=3 y^2,
\]
then $\mathcal{V}=(0,\infty)^3$ is a forward invariant set. As $(dy/d\xi)_{|y=1}<0$ on $\mathcal{V}$, the region $\mathcal{V}_*=(0,\infty)\times(0,1)\times(0,\infty)$ is also forward invariant, which in particular using~\eqref{xeq} proves the claim (i) in the theorem.

We now show that $z\to 1$ as $\xi\to\infty$ along any orbit of~\eqref{newdynsys}. In fact, if $\gamma=(u(\xi),y(\xi),z(\xi))$ is an orbit in the region $\{z>1\}$, then $z(\xi)$ is decreasing, while if $\gamma\in\{z<1\}$, then $z(\xi)$ is increasing. Hence the $z$-component of any orbit must have a limit. If this limit is not 1, then it must hold $y(\xi)\to0$, as $\xi\to\infty$, but this is impossible because $(dy/d\xi)_{|y=0}=a>0$. We conclude that the $\omega$-limit set of any orbit $\gamma$ coincides with the $\omega$-limit set of the projection of $\gamma$ on the surface $z=1$. Hence we can restrict to study the 2-dimensional dynamical system induced by~\eqref{newdynsys} on $z=1$, which is
\begin{equation}\label{newdynsys2}
\frac{du}{d\xi}=-u(1-2y)y,\quad\frac{dy}{d\xi}=(a+by+y^2)(1-y)-u^2y.
\end{equation}
Let $\mathcal{U}=(0,\infty)^2$
and let $\Gamma$ be the set of all orbits $\gamma$ of~\eqref{newdynsys2} such that $\gamma\cap\mathcal{U}$ is not empty. Since $\mathcal{U}$ is forward invariant, then for all $\gamma\in\Gamma$ there exists $\xi_0\in\R$ such that $\gamma(\xi)\in\mathcal{U}$, for all $\xi>\xi_0$. As we are only interested in positive solutions of~\eqref{generalssSeth}, we can restrict ourselves to study the behavior of the orbits $\gamma\in\Gamma$.


We now show that any orbit $\gamma\in\Gamma$ must eventually be trapped in the region $\mathcal{U}_*=(0,\infty)\times(0,1)$.
If not, and since $\mathcal{U}_*$ is forward invariant, then there would exist $\gamma\in\Gamma$, $\gamma(\xi)=(u(\xi),y(\xi))$, such that $\gamma\cap\mathcal{U}_*$ is empty. Thus $y(\xi)>1$ for all sufficiently large $\xi\in\R$ and so $u(\xi)\to\infty$ as $\xi\to\infty$. In particular there exists $\xi_0\in\R$ such that $u(\xi)>1$, for all $\xi>\xi_0$. Hence $dy/d\xi\leq -y$, that is $y(\xi)\leq C \exp(-\xi)$, for $\xi>\xi_0$, where $C>0$ is a numerical constant, which leads to the contradiction that $y(\xi)<1$ for $\xi$ large.

The dynamical system~\eqref{newdynsys2} has two fixed points:
\[
P: (u_P,y_P)=\frac{1}{2}\left(\sqrt{1+4a+2b},1\right),\quad
Q: (u_Q,y_Q)=(0,1).
\]
The fixed point $P$ is an hyperbolic sink and corresponds to the self-similar solution~\eqref{sethexact}, while the fixed point $Q$ is an hyperbolic saddle, with the unstable direction pointing towards the interior of $\mathcal{U}_*$. In particular, $Q$ is the source of one only interior orbit. It is easy to show that this orbit corresponds to the solutions of~\eqref{Setheqs} with regular center, but we shall not make use of this fact.
The claim (ii) in the theorem has now been reduced to the statement that the fixed point $P$ is the $\omega$-limit of all orbits $\gamma\in\Gamma$ of~\eqref{newdynsys2}. To prove this we begin by showing that there are no periodic orbits within $\mathcal{U}_*$. Let $v(u,y)$ be the vector field in the right hand side of~\eqref{newdynsys2}. We have
\[
\mathrm{div}\,v(u,y)=b-a+y-2by-u^2-y^2.
\]
If $\lambda\leq 2\mu$ we estimate
\[
\mathrm{div}\,v(u,y)\leq b-a+y(1-2b)=-\frac{2\lambda}{\lambda+2\mu}+\frac{\lambda-2\mu}{\lambda+2\mu}y<0,\quad (u,y)\in \mathcal{U}_*.
\]
If $\lambda>2\mu$ we estimate
\[
\mathrm{div}\,v(u,y)\leq 1+b-a=\frac{2\mu-\lambda}{\lambda+2\mu}<0,\quad (u,y)\in \mathcal{U}_*.
\]
Hence $\mathrm{div}\,v(u,y)<0$ for all $(u,y)\in \mathcal{U}_*$ and thus, by the Bendixson-Dulac theorem, there are no periodic orbits within $\mathcal{U}_*$. Next define
\[
D=\{(u,y):(u-u_P)^2+(y-y_P)^2<1/2\}.
\]
We now show that the open region $D\cap \mathcal{U}_*$ 
is forward invariant (see Figure~\ref{figure}). As we already know that $\mathcal{U}_*$ is forward invariant, it suffices to prove that
\[
\left(\frac{d}{d\xi}[(u-u_P)^2+(y-y_P)^2]\right)_{\partial D\cap\, \mathcal{U}_*}<0.
\]
A simple computation shows that
\[
\frac{d}{d\xi}[(u-u_P)^2+(y-y_P)^2]=2y(u-u_P)^2(y-y_P)-(y-y_P)^2(4a+y(2y+2b-1)).
\]
Using $a\geq 1$ and $b\leq 1$ we find
\[
4a+y(2y+2b-1)\geq 4a-\frac{(2b-1)^2}{8}\geq 4a-\frac{1}{8}\geq \frac{31}{8}.
\]
Hence 
\[
\left(\frac{d}{d\xi}[(u-u_P)^2+(y-y_P)^2]\right)_{\partial D\cap \mathcal{U}_*}\leq 2(1/2)^{3/2}-\frac{31}{8}(1/2)<0.
\]
We now prove that each orbit $\gamma\in\Gamma$ must eventually be trapped within $D\cap\mathcal{U}_*$. If not, and since $D\cap\mathcal{U}_*$ is forward invariant, then there would exist an orbit $\gamma\in\Gamma$, $\gamma(\xi)=(u(\xi),y(\xi))$, such that, for all sufficiently large $\xi$,  either $0<u(\xi)<u_P-1/2$ or $u(\xi)>u_P+1/2$, see Figure~\ref{figure}.  In the first case the orbit is positively bounded, hence its $\omega$-limit set must exist. However, due to the absence of periodic orbits and the local properties of  the fixed point $Q$ described above, the Poincar\'e-Bendixson theorem entails that no point in the region $(0,u_P-1/2)\times (0,1)$ can be a $\omega$ limit point, hence no orbit can be trapped in this region. In the second case, i.e., $u(\xi)>u_P+1/2$ for all sufficiently large $\xi$, we show that $y(\xi)<1/2$, for $\xi$ large enough. In fact,
\begin{align*}
\left(\frac{dy}{d\xi}\right)_{u>u_p+1/2,1/2\leq y<1}&\leq \frac{1}{2}(1+a+b)-\frac{1}{2}(u_p+1/2)^2\\
&=\frac{1}{2}(1/2+b/2-u_p)\\
&\leq \frac{1}{4}(1+b-\sqrt{5+2b})< 0,\ \text{for all $0\leq b\leq 1$}.
\end{align*}
Thus each orbit $\gamma\in\Gamma$ that satisfies $u(\xi)>u_P+1/2$ for large $\xi$ must eventually be trapped in the subregion where $0<y<1/2$. Hence $du/d\xi<0$ for large $\xi$, which implies that the orbit is positively bounded. Thus its $\omega$-limit set should be non-empty, but, again due to Poincar\'e-Bendixson's theorem, this is not possible. We conclude that, as claimed,  each orbit intersecting $\gamma\in\Gamma$ must eventually be trapped within $D\cap\mathcal{U}_*$ and so its $\omega$-limit set must lie in this region. As the only possible $\omega$-limit point in $D\cap\mathcal{U}_*$ is the fixed point $P$, the proof is complete.
\end{proof}

\begin{figure}[ht!]
	\begin{center}
		\psfrag{P}[cc][cc][1.0][0]{$\mathbf{\mathrm{P}}$}
		\psfrag{Q}[cc][cc][1.0][0]{$\mathbf{\mathrm{Q}}$}
		\psfrag{u-}[cc][cc][1.0][0]{$u_{P}-\frac{1}{2}$}
		\psfrag{u+}[cc][cc][1.0][0]{$u_{P}+\frac{1}{2}$}
		\psfrag{b}[cc][cc][1.0][0]{$\partial D\cap\mathcal{U}_*$}
		\psfrag{y}[cc][cc][1.0][0]{$y$}
		\psfrag{y1}[cc][cc][1.0][0]{$y=1$}
		\psfrag{u}[cc][cc][1.0][0]{$u$}
		\includegraphics[width=0.6\textwidth, trim = 0cm 0cm 0cm 0cm]{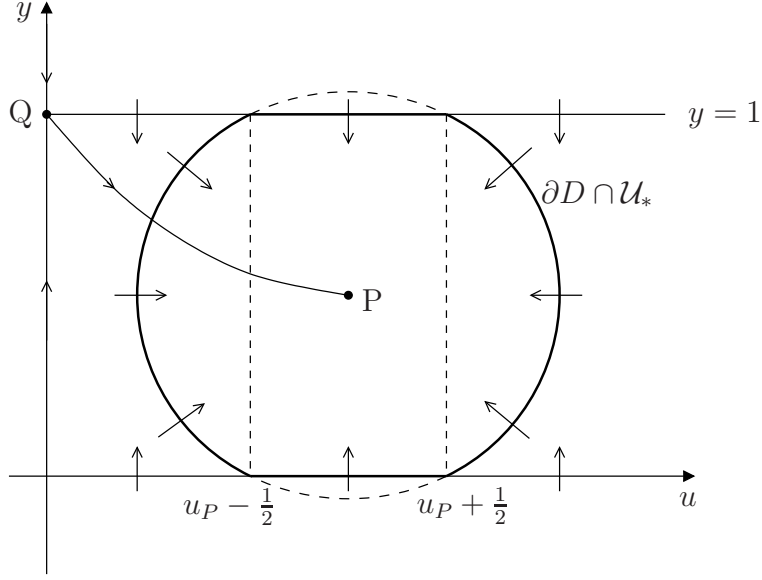}
	\end{center}
	\caption{All orbits 
		of the dynamical system~\eqref{newdynsys} entering the region $\mathcal{U}=(0,\infty)^2$ are eventually trapped in the region $D\cap\mathcal{U}_*$
	}
	\label{figure}
\end{figure}
\begin{corollary}\label{coro}
For all regular solutions of~\eqref{Setheqs} there holds
\[
\bar{p}_\mathrm{rad}(r):=\widehat{p}_\mathrm{rad}(\delta(r),\eta(r))\to -p_0=-\frac{3\lambda+2\mu}{2},\quad \text{as $r\to\infty$.}
\]
\end{corollary}
With the results of Theorem~\ref{cauchyprob} and Corollary~\ref{coro} at hand we can now discuss the existence of static, spherically symmetric, self-gravitating $n$-body distributions in the Seth model. For this we shall need the following simple lemma.
\begin{lemma}\label{lemmino}
For the Seth model, the conditions $p_\mathrm{rad}(r)=0$ and $\delta(r)\geq0$ hold simultaneously if and only if
\begin{equation}\label{etarho}
\eta(r)\leq \left(\frac{3\lambda+2\mu}{2\lambda}\right)^{3/2},\quad\delta(r)=\eta(r)^{2/3}\sqrt{\frac{2\lambda}{\lambda+2\mu}}\left(\frac{3\lambda+2\mu}{2\lambda}-\eta(r)^{2/3}\right)^{1/2}.
\end{equation}
Moreover at all radii $r>0$ such that $p_\mathrm{rad}(r)=0$ and $\delta(r)\geq0$ there holds
\begin{equation}
p'_\mathrm{rad}(r)=\frac{2\mu}{r}\,\frac{3\lambda+2\mu}{\lambda+2\mu}\,(\eta(r)^{2/3}-1)-\frac{\mathcal{K}}{r^2}\,m(r)\delta(r).\label{plinha}
\end{equation}
In particular if $0<\eta(r)\leq 1$, then $p'_\mathrm{rad}(r)<0$, while if $\delta(r)=0$, then $p'_\mathrm{rad}(r)>0$.
\end{lemma}
\begin{proof}
The first statement follows by solving $\widehat{p}_\mathrm{rad}(\delta(r),\eta(r))=0$ in terms of $\delta(r)$. As to~\eqref{plinha}, we have, for all $r>0$ such that $p_\mathrm{rad}(r)=0$, 
\[
p'_\mathrm{rad}(r)=\frac{2p_\mathrm{tan}(r)}{r}-\rho(r)\frac{m(r)}{r^2}=2\frac{\widehat{p}_\mathrm{tan}(\delta(r),\eta(r))}{r}-\mathcal{K}\delta(r)\frac{m(r)}{r^2}.
\]
Substituting the formula~\eqref{etarho} for $\delta(r)$ in $\widehat{p}_\mathrm{tan}(\delta(r),\eta(r))$ completes the proof of~\eqref{plinha}. 
\end{proof}
We can now prove the existence of static, self-gravitating balls in the Seth model.
\begin{theorem}\label{balltheo}
Let $\rho_c>0$ be given. A unique static, self-gravitating elastic ball $(\rho,p_\mathrm{rad},p_\mathrm{tan})$ in the Seth model with reference density $\mathcal{K}$ and central density $\rho(0)=\rho_c$ exists if and only if $\rho_c>\mathcal{K}$. Moreover, if $r_1>0$ is the radius of the ball, then $\rho(r_1)>0$ and the following estimates hold:
\begin{equation}\label{estr}
\sqrt{\left(\frac{2\lambda}{3\lambda+2\mu}\right)}\left(\frac{3M}{4\pi\mathcal{K}}\right)^{1/3}< r_1< \left(\frac{3M}{4\pi\mathcal{K}}\right)^{1/3},\quad M< \frac{4\pi}{3}\rho_c r_1^3,
\end{equation}
where $M=m(r_1)>0$ is the total mass of the ball.
\end{theorem}
\begin{proof}
Let $(\rho,p_\mathrm{rad},p_\mathrm{tan})$ be a static, self-gravitating elastic ball in the Seth model satisfying $\rho(0)=\rho_c$. By definition it must hold that $p_\mathrm{rad}(0)>0$. Using $\eta(r)\to\delta_c=\rho_c/\mathcal{K}$ as $r\to 0^+$, we have
\[
p_\mathrm{rad}(r)\to p_0\left(\left(\frac{\rho_c}{\mathcal{K}}\right)^{2/3}-1\right),\quad\text{as $r\to 0^+$,}
\]
which shows that the central radial pressure is positive if and only if $\rho_c>\mathcal{K}$. This concludes the proof of the ``only if" portion of the theorem. To prove the logically opposite statement, 
let $(\delta,\eta,m)\in C^1((0,\infty))$ be the unique regular solution of~\eqref{Setheqs} satisfying $\delta(0)=\eta(0)=\delta_c=\rho_c/\mathcal{K}>1$ and $m(0)=0$. Define $\bar{\rho}(r)=\mathcal{K}\,\delta(r)$, $\bar{p}_\mathrm{rad}(r)=\widehat{p}_\mathrm{rad}(\delta(r),\eta(r))$ and $\bar{p}_\mathrm{tan}(r)=\widehat{p}_\mathrm{tan}(\delta(r),\eta(r))$. It follows by Corollary~\ref{coro} that there exists a unique $r_1>0$ such that $\bar{p}_\mathrm{rad}(r_1)=0$ and $\bar{p}_\mathrm{rad}(r)>0$, for $r\in [0,r_1)$. Moreover 
\[
y(0)=\frac{\delta(0)}{\eta(0)}=1,
\]
hence by Theorem~\ref{cauchyprob}(i)
\begin{equation}\label{cavolinola}
y(r)=\frac{\delta(r)}{\eta(r)}<1,\quad\text{for all $r>0$}
\end{equation}
and thus
\begin{align*}
\bar{p}_\mathrm{tan}(r)&=\bar{p}_\mathrm{rad}(r)+\bar{p}_\mathrm{tan}(r)-\bar{p}_\mathrm{rad}(r)\\
&=\bar{p}_\mathrm{rad}(r)+\mu\eta(r)^{2/3}(1-y(r))
> \bar{p}_\mathrm{rad}(r)\geq0,\quad \text{for $r\in (0,r_1]$}.
\end{align*}
Hence $(\rho,p_\mathrm{rad},p_\mathrm{tan})=(\bar{\rho},\bar{p}_\mathrm{rad},\bar{p}_\mathrm{tan})\mathbb{I}_{r\leq r_1}$ defines a self-gravitating ball of radius $r_1$ and mass $M=m(r_1)$, which is uniquely determined by $\rho_c$ due to the uniqueness of the solution $(\delta,\eta,m)$ of~\eqref{generalssSeth} with data $\delta(0)=\eta(0)=\delta_c$, $m(0)=0$.
 If $\delta(r_1)=0$, then $\lim_{r\to r_1^+}p'_\mathrm{rad}(r)>0$, see Lemma~\ref{lemmino}, which is clearly impossible. Hence $\delta(r_1)>0$, i.e., $\rho(r_1)>0$, must hold. Thus
\[
\eta(r_1)=\frac{M}{\frac{4\pi}{3}\mathcal{K} r_1^3}<\left(\frac{3\lambda+2\mu}{2\lambda}\right)^{3/2},
\]
which gives the lower bound in~\eqref{estr}. As to the upper bound, we have
\begin{align*}
0&=\widehat{p}_\mathrm{rad}(\delta(r_1),\eta(r_1))=\lambda\eta(r_1)^{2/3}+\frac{\lambda+2\mu}{2}\eta(r_1)^{-4/3}\delta(r_1)^2-p_0\\
&=\eta^{2/3}(r_1)[(p_0-\lambda)(y(r_1)^2-1)+p_0(1-\eta(r_1)^{-2/3})].
\end{align*}
Using $p_0>\lambda$ and $y(r_1)<1$, see~\eqref{cavolinola}, we obtain $1-\eta(r_1)^{-2/3}>0$, i.e., $\eta(r_1)>1$, which gives the upper bound in~\eqref{estr}. Finally, Theorem~\ref{cauchyprob}(i) gives $\eta(r_1)< \eta(0)=\rho_c/\mathcal{K}$, hence
\[
M=m(r_1)=\frac{4\pi}{3}\mathcal{K}\,r_1^3\,\eta(r_1)< \frac{4\pi}{3}\rho_cr_1^3.
\]
\end{proof}
We now use a similar argument to prove the existence of single shells of Seth elastic matter.
\begin{theorem}\label{singleshelltheo}
A unique single static, self-gravitating elastic shell $(\rho,p_\mathrm{rad},p_\mathrm{tan})$ in the Seth model with reference parameters $\mathcal{K},\mathcal{S}$ and inner radius $r_0>0$ exists if and only if
\begin{equation}\label{necShell}
r_\mathrm{min}:=\left(\frac{2\lambda}{3\lambda+2\mu}\right)^{1/2}\mathcal{S}\leq r_0<\mathcal{S}.
\end{equation}
Moreover $\rho(r_0)=0$ if and only if $r_0=r_\mathrm{min}$, while the outer radius $r_1$ of the shell satisfies $\rho(r_1)>0$ and
\begin{equation}\label{estrshell}
\sqrt{\frac{2\lambda}{3\lambda+2\mu}}\left(\mathcal{S}^3+\frac{3M}{4\pi\mathcal{K}}\right)^{1/3}<r_1<\left(\mathcal{S}^3+\frac{3M}{4\pi\mathcal{K}}\right)^{1/3},
\end{equation}
where $M$ is the total mass of the shell.
\end{theorem}
\begin{proof}
If a single shell with inner radius $r_0>0$ exists, then it must hold
\begin{subequations}\label{necshelltemp}
\begin{equation}\label{cavolinaccio1}
\eta(r_0)=\left(\frac{\mathcal{S}}{r_0}\right)^3,\quad \delta(r_0)=\left(\frac{\mathcal{S}}{r_0}\right)^{2}\sqrt{\frac{2\lambda}{\lambda+2\mu}}\left(\frac{3\lambda+2\mu}{2\lambda}-\left(\frac{\mathcal{S}}{r_0}\right)^{2}\right)^{1/2}\quad (\text{i.e.,  $p_\mathrm{rad}(r_0)=0$}),
\end{equation}
and
\begin{equation}\label{cavolinaccio}
p_\mathrm{rad}'(r_0)=\frac{2\mu}{r_0}\,\frac{3\lambda+2\mu}{\lambda+2\mu}\,\left[\left(\frac{\mathcal{S}}{r_0}\right)^{2/3}-1\right]>0,\quad\text{or}\quad \left(r_0=\mathcal{S}\ \text{and}\ p_\mathrm{rad}''(\mathcal{S})\geq 0\right).
\end{equation}
\end{subequations}
A simple computation shows that, if~\eqref{cavolinaccio1} holds with $r_0=\mathcal{S}$, then $p_\mathrm{rad}''(\mathcal{S})=-4\pi\mathcal{K}$, hence the second condition in~\eqref{cavolinaccio} is never satisfied together with~\eqref{cavolinaccio1}. Thus~\eqref{necshelltemp} implies
\[
1<\frac{\mathcal{S}}{r_0}\leq \left(\frac{3\lambda+2\mu}{2\lambda}\right)^{1/2},
\]
which in turn is equivalent to~\eqref{necShell}. Now, given $r_0$ satisfying~\eqref{necShell}, define $\eta(r_0)$ and $\delta(r_0)$ as in~\eqref{cavolinaccio1}, 
so that $p_\mathrm{rad}(r_0)=0$ and $p'_\mathrm{rad}(r_0)>0$. Let $(\delta,\eta,m)\in C^1((r_0,\infty))$ be the unique regular solution of~\eqref{Setheqs} with initial data $(\delta(r_0),\eta(r_0), m(r_0)=0)$.  Define $\bar{\rho}(r)=\mathcal{K}\,\delta(r)$, $\bar{p}_\mathrm{rad}(r)=\widehat{p}_\mathrm{rad}(\delta(r),\eta(r))$ and $\bar{p}_\mathrm{tan}(r)=\widehat{p}_\mathrm{tan}(\delta(r),\eta(r))$. It follows by Corollary~\ref{coro} that there exists a unique $r_1>0$ such that $\bar{p}_\mathrm{rad}(r_1)=0$ and $\bar{p}_\mathrm{rad}(r)>0$, for $r\in (r_0,r_1)$. Moreover, using $\eta(r_0)>1$, and letting $y(r_0)=\delta(r_0)/\eta(r_0)$, we have
\begin{align*}
0&=\bar{p}_\mathrm{rad}(r_0)=\widehat{p}_\mathrm{rad}(\delta(r_0),\eta(r_0))=\lambda\eta(r_0)^{2/3}+\frac{\lambda+2\mu}{2}\eta(r_0)^{-4/3}\delta(r_0)^2-p_0\\
&=\eta^{2/3}(r_0)[(p_0-\lambda)(y(r_0)^2-1)+p_0(1-\eta(r_0)^{-2/3})]>\eta^{2/3}(r_0)(p_0-\lambda)(y(r_0)^2-1),
\end{align*}
hence $0<y(r_0)<1$. By Theorem~\ref{cauchyprob}(i)
\begin{equation}\label{cavolino}
y(r)=\frac{\delta(r)}{\eta(r)}<1,\quad\text{for all $r>r_0$}
\end{equation}
and thus
\begin{align*}
\bar{p}_\mathrm{tan}(r)&=\bar{p}_\mathrm{rad}(r)+\bar{p}_\mathrm{tan}(r)-\bar{p}_\mathrm{rad}(r)\\
&=\bar{p}_\mathrm{rad}(r)+\mu\eta(r)^{2/3}(1-y(r))
> \bar{p}_\mathrm{rad}(r)\geq0,\quad \text{for $r\in [r_0,r_1]$}.
\end{align*}
Hence $(\rho,p_\mathrm{rad},p_\mathrm{tan})=(\bar{\rho},\bar{p}_\mathrm{rad},\bar{p}_\mathrm{tan})\mathbb{I}_{r_0\leq r\leq r_1}$ defines a self-gravitating shell supported in the interval $[r_0,r_1]$ and with mass $M=m(r_1)$, which is uniquely determined by $r_0$ due to the uniqueness of the solution $(\delta,\eta,m)$ of~\eqref{generalssSeth} with initial data $(\delta(r_0),\eta(r_0),m(r_0))$. The proof of~\eqref{estrshell} is identical to that of~\eqref{estr}. In fact if $\rho(r_1)=0$, then Lemma~\ref{lemmino} gives $\lim_{r\to r_1^+}p'_\mathrm{rad}(r)>0$, which is not possible. Hence
\[
\delta(r_1)>0,\quad\text{i.e.}\quad \eta(r_1)=\left(\frac{\mathcal{S}}{r_1}\right)^3+\frac{3M}{4\pi\mathcal{K}r_1^3}<\left(\frac{3\lambda+2\mu}{2\lambda}\right)^{3/2},
\] 
which gives the lower bound on $r_1$. As to the upper bound, we have
\begin{align*}
0&=\widehat{p}_\mathrm{rad}(\delta(r_1),\eta(r_1))=\lambda\eta(r_1)^{2/3}+\frac{\lambda+2\mu}{2}\eta(r_1)^{-4/3}\delta(r_1)^2-p_0\\
&=\eta^{2/3}(r_1)[(p_0-\lambda)(y(r_1)^2-1)+p_0(1-\eta(r_1)^{-2/3})].
\end{align*}
Using $p_0>\lambda$ and $y(r_1)<1$, see~\eqref{cavolino}, we obtain $1-\eta(r_1)^{-2/3}>0$, i.e., $\eta(r_1)>1$, which gives the upper bound in~\eqref{estr}. 
\end{proof}
The next theorem generalizes Theorem~\ref{singleshelltheo} to the case when the region surrounded by the shell of Seth matter is not vacuum. 
\begin{theorem}\label{M+shell}
Let $n\in\mathbb{N}$, $n>1$, and $\mathfrak{B}_1+\dots+\mathfrak{B}_n$ be a $n$-body static, spherically symmetric, self-gravitating matter distribution. Assume that $\mathfrak{B}_n$ is a shell of Seth elastic matter with Lam\'e coefficients $\lambda,\mu$ and with reference parameters $\mathcal{K},\mathcal{S}$. Then the inner radius $r_{2n-2}$ of $\mathfrak{B}_n$ satisfies
\begin{equation}\label{necShell2}
r_\mathrm{min}:=\left(\frac{2\lambda}{3\lambda+2\mu}\right)^{1/2}\mathcal{S}\leq r_{2n-2}<\mathcal{S}.
\end{equation}
In particular $\mathcal{S}>r_{2n-3}$ must hold.
Conversely,  if $r_\mathrm{min}>r_{2n-3}$,
then
there exists $r_\mathrm{max}\in(r_\mathrm{min},\mathcal{S})$ such that for all given $r_{2n-2}\in [r_\mathrm{min},r_\mathrm{max})$ there exists a unique $r_{2n-1}>r_{2n-2}$ and a unique single shell $\mathfrak{B}_{n}=(\rho,p_\mathrm{rad},p_\mathrm{tan})$ of the given Seth elastic material supported in the interval $[r_{2n-2},r_{2n-1}]$, such that $\mathfrak{B}_1+\dots+\mathfrak{B}_{n}$ is a static, spherically symmetric, self-gravitating $n$-body matter distribution.  Moreover $\rho(r_{2n-2})=0$ if and only if $r_{2n-2}=r_\mathrm{min}$, while the outer radius $r_{2n-1}$ of the elastic shell $\mathfrak{B}_{n}$ satisfies $\rho(r_{2n-1})>0$ and 
\begin{equation}\label{estrshell2}
\sqrt{\frac{2\lambda}{3\lambda+2\mu}}\left(\mathcal{S}^3+\frac{3M_{n}}{4\pi\mathcal{K}}\right)^{1/3}<r_{2n-1}<\left(\mathcal{S}^3+\frac{3M_{n}}{4\pi\mathcal{K}}\right)^{1/3},
\end{equation}
where $M_{n}$ is the mass of $\mathfrak{B}_{n}$.
\end{theorem}
\begin{proof}
At the inner radius $r_{2n-2}$ of $\mathfrak{B}_n$ it must hold  
\[
\eta(r_{2n-2})=\left(\frac{\mathcal{S}}{r_{2n-2}}\right)^3,\quad \delta(r_{2n-2})=\left(\frac{\mathcal{S}}{r_{2n-2}}\right)^{2}\sqrt{\frac{2\lambda}{\lambda+2\mu}}\left(\frac{3\lambda+2\mu}{2\lambda}-\left(\frac{\mathcal{S}}{r_{2n-2}}\right)^{2}\right)^{1/2},
\]
so that $p_\mathrm{rad}(r_{2n-2})=0$, and
\[
p'_\mathrm{rad}(r_{2n-2})=\frac{2\mu}{r_{2n-2}}\,\frac{3\lambda+2\mu}{\lambda+2\mu}\left[\left(\frac{\mathcal{S}}{r_{2n-2}}\right)^2-1\right]-\frac{\mathcal{K}\sum_{k=1}^{n-1}M_k}{r_{2n-2}^2}\,\delta(r_{2n-2}):=F(r_{2n-2})\geq0,
\]
from which the  necessary condition~\eqref{necShell2} follows immediately. Since $F(r_\mathrm{min})>0$, then by continuity there exists $r_\mathrm{max}\in(r_\mathrm{min},\mathcal{S})$ such that $F(r)>0$, for all $r\in[r_\mathrm{min},r_\mathrm{max})$. Hence, provided $\mathcal{S}$ is so that $r_\mathrm{min}>r_{2n-3}$ and choosing $r_{2n-2}\in[r_\mathrm{min},r_\mathrm{max})$, the (unique) solution of~\eqref{generalssSeth} with initial data $(\eta(r_{2n-2}),\delta(r_{2n-2}),m(r_{2n-2})=M_1+\dots M_{n-1})$ gives rise to a shell supported in $[r_{2n-2},r_{2n-1}]$, where $r_{2n-1}$ is the first radius at which the radial pressure vanishes (and which exists by Corollary~\ref{coro}). Finally, as $r_\mathrm{min}\leq r_{2n-2}<\mathcal{S}$, the proof of~\eqref{estrshell2} is identical to the proof of~\eqref{estrshell}.
\end{proof} 

The previous theorem can be used to construct spherically symmetric, static, self-gravitating $n$-body matter distributions in the Seth model recursively, which is particular useful for, e.g., numerical simulations. In fact, suppose that $\mathfrak{B}_1$ is a ball of Seth elastic matter with Lam\'e coefficients $(\lambda^{(1)},\mu^{(1)})$ and with reference density $\mathcal{K}^{(1)}$. The existence of such ball is proved in Theorem~\ref{balltheo}. Alternatively, $\mathfrak{B}_1$ could be a single shell of the given Seth material and with reference parameters $(\mathcal{K}^{(1)},\mathcal{S}^{(1)})$, see Theorem~\ref{singleshelltheo}. The (outer) radius $r_1$ of $\mathfrak{B}_1$ depends on the Lam\'e coefficients and the reference parameters of $\mathfrak{B}_1$, that is 
\[
r_1=r_1(\lambda^{(1)},\mu^{(1)},\mathcal{K}^{(1)},\mathcal{S}^{(1)}).
\]
By Theorem~\ref{M+shell} we can add a shell $\mathfrak{B}_2$ of Seth matter around $\mathfrak{B}_1$ if the Lam\'e coefficients $(\lambda^{(2)},\mu^{(2)})$ and the reference inner radius $\mathcal{S}^{(2)}$ of $\mathfrak{B}_2$ satisfy
\[
\left(\frac{2\lambda^{(2)}}{3\lambda^{(2)}+2\mu^{(2)}}\right)^{1/2}\mathcal{S}^{(2)}>r_{1}.
\]
The outer radius of $\mathfrak{B}_2$ depends on the Lam\'e coefficients and the reference parameters of $\mathfrak{B}_1,\mathfrak{B}_2$, that is 
\[
r_3=r_3(\lambda^{(1)},\lambda^{(2)},\mu^{(1)},\mu^{(2)},\mathcal{K}^{(1)},\mathcal{K}^{(2)},\mathcal{S}^{(1)},\mathcal{S}^{(2)}).
\]
Similarly, we can place a shell $\mathfrak{B}_3$ of Seth matter around $\mathfrak{B}_1+\mathfrak{B}_2$ if the Lam\'e coefficients $(\lambda^{(3)},\mu^{(3)})$ and the reference parameters $(\mathcal{K}^{(3)},\mathcal{S}^{(3)})$ of $\mathfrak{B}_3$ satisfy 
\[
\left(\frac{2\lambda^{(3)}}{3\lambda^{(3)}+2\mu^{(3)}}\right)^{1/2}\,\mathcal{S}^{(3)}>r_{3},
\]
and so on. 
The general sufficient condition to add a shell $\mathfrak{B}_{j+1}$ of Seth matter around a $j$-body matter distribution  $\mathfrak{B}_1+\dots+\mathfrak{B}_{j}$ of Seth matter, $j\geq 1$, is that the Lam\'e coefficients $(\lambda^{(j+1)},\mu^{(j+1)})$ and the reference parameters $(\mathcal{K}^{(j+1)},\mathcal{S}^{(j+1)})$ of $\mathfrak{B}_{j+1}$ satisfy
\[
\left(\frac{2\lambda^{(j+1)}}{3\lambda^{(j+1)}+2\mu^{(j+1)}}\right)^{1/2}\mathcal{S}^{(j+1)}>r_{2j-1},
\]
where
\[
r_{2j-1}=r_{2j-1}(\lambda^{(1)},\dots,\lambda^{(j)},\mu^{(1)},\dots,\mu^{(j)},\mathcal{K}^{(1)},\dots,\mathcal{K}^{(j)},\mathcal{S}^{(1)},\dots,\mathcal{S}^{(j)})
\]
is the outer radius of the shell $\mathfrak{B}_j$. We remark that this result applies in particular when the bodies are made of the same Seth elastic material, i.e., when they all have the same Lam\'e coefficients $\lambda,\mu$. In this case Theorem~\ref{M+shell} entails that a static, spherically symmetric, $n$-body matter distribution exists provided the reference inner radii $\mathcal{S}^{(1)},\dots,\mathcal{S}^{(n)}$ of the bodies satisfy
 \[
S^{(j)}>\left(\frac{3\lambda+2\mu}{2\lambda}\right)^{1/2}r_{2j-3}(\lambda,\mu,\mathcal{K}^{(1)},\dots,\mathcal{K}^{(j-1)},\mathcal{S}^{(1)},\dots,\mathcal{S}^{(j-1)}),\quad j=2,\dots, n.
\]
In particular, the bodies are allowed to have the same reference density, i.e.,
 \[
\mathcal{K}^{(1)}=\dots=\mathcal{K}^{(n)}=\mathcal{K}.
\]

\section*{Acknowledgements}
A.~A. is supported by the project (GPSEinstein) PTDC/MAT-ANA/1275/2014,
by CAMGSD, Instituto Superior T{\'e}cnico, through FCT/Portugal 
UID/MAT/04459/2013, and by the FCT Grant No. SFRH/BPD/85194/2012. 
Furthermore, A.~A.  thanks the Department of Mathematics at Chalmers University, Sweden, for the very kind hospitality. 

\begin{appendices}

\section{Appendix: Static self-gravitating elastic matter}\label{ssg}
We begin by recalling how the stress tensor $\sigma$ of an elastic body is defined in the Lagrangian formulation of continuum mechanics. 
Let $\mathcal{B}\subset\R^3$ be a non-empty, open, bounded, connected set with smooth boundary. The region $\mathcal{B}$ is called material manifold and corresponds to the domain occupied by the body in a given reference state.
A static (i.e., time independent) deformation of the body is described by a $C^2$ map $\psi:\overline{\mathcal{B}}\to\R^3$, such that $\psi$ is injective (except possibly at the boundary) and preserves orientation, i.e., $\det F>0$, where $F=\nabla\psi$ is the deformation gradient. The region $\psi(\overline{\mathcal{B}})$ identifies the domain occupied by the matter distribution in its physical state, i.e., after it has been deformed. 
On the material manifold $\mathcal{B}$ we define two matter fields: a reference mass density $\rho_\mathrm{ref}\in C(\overline{\mathcal{B}})$ such that $\inf_\mathcal{B} \rho_\mathrm{ref}>0$ and a 2-tensor field $\Sigma\in C^1(\overline{\mathcal{B}})$, called the {\it first} Piola-Kirchhoff stress tensor, satisfying 
\begin{equation}\label{eqSIGMA}
\mathrm{DIV}\,\Sigma\,(X)=G\rho_\mathrm{ref}(X)\int_\mathcal{B}\rho_\mathrm{ref}(Y)\frac{\psi(X)-\psi(Y)}{|\psi(X)-\psi(Y)|^3}\,dY,\quad X\in \mathcal{B}.
\end{equation}
Here $G=0$ when the self-induced gravitational force acting on the body is neglected, otherwise $G$ equals Newton's gravitational constant, see~\cite{BS,tomsimo}. The body is called elastic if $\Sigma(X)=\widehat{\Sigma}(X,F(X))$, for all $X\in\mathcal{B}$. For elastic bodies, the PDE~\eqref{eqSIGMA} transforms into a second order PDE on the deformation map, which has to be supplemented with appropriate boundary conditions, depending on the specific problem under study. 
For instance, the zero-traction boundary conditions describing an isolated body surrounded by vacuum are the non-linear Neumann-type boundary conditions
\begin{equation}\label{bcmat}
\widehat{\Sigma}(X,F(X))\cdot N=0\quad\text{on } \partial\mathcal{B},
\end{equation}
where $N$ is the outward unit normal vector field on the boundary $\partial\mathcal{B}$,  see~\cite{BS}.
The mass density $\rho$ and the Cauchy stress tensor $\sigma$ at the point $x=\psi(X)$ in the physical (deformed) state of the body are given by
\begin{equation}\label{piola}
\rho(\psi(X))=\frac{\rho_\mathrm{ref}(X)}{\det F(X)},\quad
\sigma(\psi(X))=\left(\frac{\widehat{\Sigma}(X,F(X))\cdot F^T(X)}{\det F(X)}\right).
\end{equation}
The definition of $\rho$ ensures that the mass of any subregion of the body is preserved by the deformation, while the definition of $\sigma$ ensures that~\eqref{generaleq} is satisfied for $x\in\Omega:=\psi(\mathcal{B})$. Moreover, provided the boundary of $\Omega$ is sufficiently regular, so that the normal field $n$ on $\partial\Omega$ is well defined, then the boundary conditions~\eqref{bcmat} transform into $\sigma\cdot n=0$, which in the spherically symmetric case is equivalent to the vanishing of the radial pressure on the physical boundary. 


An elastic body is called homogeneous if $\rho_\mathrm{ref}(X)=\mathcal{K}>0$ (constant) and $\Sigma(X)=\widehat{\Sigma}(F(X))$. An homogeneous elastic body is said to be isotropic if $\widehat{\Sigma}(A\cdot F)=\widehat{\Sigma}(F)$, for all $A\in SO(3)$, and frame indifferent if $\widehat{\Sigma}(A\cdot F)=A\cdot\widehat{\Sigma}(F)$, for all $A\in SO(3)$. In the following, an elastic body which is homogeneous, isotropic and frame indifferent will be called perfect elastic body for short.

The reference state of a perfect elastic body is called a natural state if $\widehat{\Sigma}(\mathbb{I})=0$, where $\mathbb{I}$ denotes the $3\times 3$ identity matrix, otherwise the reference state is said to be pre-stressed. 
  
An homogeneous elastic body is called hyperelastic if there exists a function $W$ such that $\widehat{\Sigma}(F)=\partial_FW(F)$. The function $W$ is called stored energy function of the body; without loss of generality, it is assumed that $W(\mathbb{I})=0$. 
In the case of perfect hyperelastic bodies the stored energy function can be written in the form $W(F)=\Phi(\lambda_1,\lambda_2,\lambda_3)$, where $\lambda_i$, $i=1,2,3$, are the principal stretches, that is to say, $\lambda_1^2,\lambda_2^2,\lambda_3^2$ are the eigenvalues of the (right) Cauchy-Green tensor
\[
C=F^T\cdot F. 
\]
For consistency with linear elasticity (i.e., Hooke's law) in the infinitesimal strain limit, the stored energy function $W(F)=\Phi(\lambda_1,\lambda_2,\lambda_3)$ of hyperelastic bodies should satisfy
\begin{equation}\label{lincom}
\frac{\partial^2\Phi}{\partial\lambda_i\partial\lambda_j}(1,1,1)=\lambda+2\mu\delta_{ij}.
\end{equation}
for some constants $\lambda,\mu$, called Lam\'e coefficients~\cite{ogden}


\subsection{Examples of elastic models}
In the following examples we denote $J=\det F=\det\nabla\psi$, and use the Almansi strain tensor $A$ and the Green strain tensor $E$  defined as
\[
A=\frac{1}{2}(\mathbb{I}-C^{-T}),\quad E=\frac{1}{2}(C-\mathbb{I}).
\]
The principal invariants of the Cauchy-Green tensor  are given by
\begin{align*}
i_1 (C) &=\text{Tr}(C) =\lambda^2_1+\lambda^2_2+\lambda^2_3, \\ 
i_2 (C)&= \frac{1}{2}\left[(\text{Tr}(C))^2-\text{Tr}(C^2)\right]  =\lambda^2_1\lambda^2_2+\lambda^2_1\lambda^2_3+\lambda^2_2\lambda^2_3, \\ 
i_3 (C)&= \det{(C)}=J^2= \lambda^2_1 \lambda^2_2 \lambda^2_3, \\ 
\end{align*}
while for the Almansi tensor we have
\begin{align*}
i_1(A)&=\frac{1}{2}\left[3-\left(\frac{1}{\lambda_1^2}+\frac{1}{\lambda_2^2}+\frac{1}{\lambda_3^2}\right)\right],\\
i_2(A)&=\frac{1}{4}\left(\frac{1}{\lambda_1^2\lambda_2^2}+\frac{1}{\lambda_1^2\lambda_3^2}+\frac{1}{\lambda_2^2\lambda_3^2}+\frac{1}{\lambda_1^2}+\frac{1}{\lambda_2^2}+\frac{1}{\lambda_3^2}\right);
\end{align*}
we shall not need the expression of $i_3(A)$.
\subsubsection*{Seth model}
The Seth model postulates a linear relation between the Cauchy-stress tensor and the Almansi tensor, that is  
\begin{equation}\label{cauchyseth}
\sigma(\psi(X))= \lambda i_1(A(X))\mathbb{I}+2\mu A(X),
\end{equation}
where the Lam\'e coefficients satisfy $\mu>0$, $3\lambda+2\mu>0$.
It follows by~\eqref{piola} that the first Piola-Kirchoff stress tensor of Seth materials is given by
\[
\widehat{\Sigma}(F)=J( \lambda i_1(A)\mathbb{I}+2\mu A)\cdot F^{-T}.
\]
The Seth model is not hyperelastic.
\subsubsection*{Saint-Venant-Kirchhoff model}
The Saint-Venant-Kirchhoff model postulates a linear relation between the {\it second} Piola-Kirchhoff stress tensor and the Green strain tensor, namely
\[
F^{-1}\cdot\widehat{\Sigma}(F)=\lambda\mathrm{Tr}(E)\mathbb{I}+2\mu E.
\]
This model is hyperelastic with stored energy function given by
\begin{align*}
 W(F) 
 =&\frac{\lambda+2\mu}{8}(i_1(C) -3)^2+\mu(i_1(C) -3)-\frac{\mu}{2}(i_2(C) -3) \nonumber\\
 =& \frac{\lambda}{8}(\sum^{3}_{i=1}\lambda^2_i -3)^2+\frac{\mu}{4}\sum^{3}_{i=1} (\lambda^2_i-1)^2,
\end{align*}
where the Lam\'e coefficients satisfy $\mu>0$, $3\lambda+2\mu>0$.
\subsubsection*{Quasi-linear Signorini model}
The first Piola-Kirchoff stress tensor of the (quasi-linear) Signorini model is given by
\[
\widehat{\Sigma}(F)=J S(A)\cdot F^{-T},\quad\text{i.e.,}\quad \sigma(\psi(X))=S(A(X)),
\]
where $S(A)$ is given in terms of the Almansi tensor by
\[
S(A)=[\lambda i_1(A)+\frac{1}{2}(\lambda+\mu)i_1(A)^2]\mathbb{I} +2[\mu -(\lambda+\mu)i_1(A)]A,
\]
where the Lam\'e coefficients $\lambda,\mu$ satisfy $\mu>0$, $9\lambda+5\mu>0$.
This model is hyperelastic with stored energy function given by
\begin{align}
W(F)&=J[\frac{1}{2}(\lambda+\mu)i_1(A)^2+\mu(1-i_1(A))]-\mu\nonumber\\
&=\frac{1}{8}\lambda_1\lambda_2\lambda_3\left[4\mu\left(\sum_{i=1}^3\lambda_i^{-2}-1\right)+(\lambda+\mu)\left(\sum_{i=1}^3\lambda_i^{-2}-3\right)^2\right]
-\mu.\label{wsig}
\end{align}


\subsubsection*{Hadamard model}
The stored energy function of Hadamard materials has the form
\begin{align}
W(F)&=\frac12(\alpha(i_1(C)-3)+\beta(i_2(C)-3)+h(i_3(C))-h(1))\nonumber\\
&=\frac{1}{2}(\alpha(\lambda_1^2+\lambda_2^2+\lambda_3^2-3)+\beta(\lambda_1^2\lambda_2^2+\lambda_1^2\lambda_3^2+\lambda_2^2\lambda_3^2-3)
+h(\lambda_1^2\lambda_2^2\lambda_3^2)-h(1)),\label{whad}
\end{align}
where $\alpha,\beta$ are constant and $h$ is a function.
The natural state condition $\partial_FW(\mathbb{I})=0$ and the compatibility conditions~\eqref{lincom} with linear elasticity are satisfied if and only if 
\[
h'(1)=-(\alpha+2\beta),\quad h''(1)=\frac{\lambda+2\mu}{2},\quad \alpha+\beta=\mu.
\] 
\subsubsection*{Incompressible materials}
All hyperelastic materials considered above are compressible. By adding the constraint $J=\lambda_1\lambda_2\lambda_3=1$ in the stored energy function, one obtains the incompressible version of each material. 
For instance, the stored energy function of the incompressible Hadamard materials is
\begin{equation}\label{wmr}
W(F)=\frac{1}{2}(\alpha(\lambda_1^2+\lambda_2^2+\lambda_3^2-3)+\beta(\lambda_1^2\lambda_2^2+\lambda_1^2\lambda_3^2+\lambda_2^2\lambda_3^2-3),\quad \lambda_1\lambda_2\lambda_3=1.
\end{equation}
Materials with the stored energy function~\eqref{wmr} are also called Mooney-Rivlin materials, see~\cite{BH}.


\subsection{Spherically symmetric elastic bodies}\label{ssem}
In this and the next section we show that the definition of homogeneous elastic bodies in Lagrangian coordinates is formally equivalent, in spherical symmetry, to Definition~\ref{elasticdef} for balls and Definition~\ref{elasticdefshell} for shells in physical space.

A deformation is spherically symmetric if there exists a function $\xi:[0,\infty)\to \R$ such that 
\[
\psi(X)=\xi(R)\frac{X}{R},\quad \text{where $R=|X|$.}
\]      
Moreover the material manifold $\mathcal{B}$ is either a ball of radius $\mathcal{T}>0$, i.e.,
\[
\mathcal{B}=\{X\in\R^3: 0\leq R<\mathcal{T}\},
\]
or a shell with inner radius $\mathcal{S}$ and outer radius $\mathcal{T}$, i.e., 
\[
\mathcal{B}=\{X\in\R^3: \mathcal{S}<R<\mathcal{T}\}.
\]
When $\mathcal{B}$ is a ball, regularity of $\psi$ up to the center requires $\xi(0)=0$, i.e., a ball cannot be deformed to a shell by a smooth deformation.

The deformation gradient $F=\nabla\psi$ has the form
\begin{equation}\label{formF}
F_{IJ}(X)=\frac{n(R)}{\tau(R)^2}\frac{X_IX_J}{R^2}+\tau(R)\left(\delta_{IJ}-\frac{X_IX_J}{R^2}\right),
\end{equation}
where
\begin{equation}\label{nR}
\tau(R)=\frac{\xi(R)}{R},\quad n(R):=\det F(R)=\tau(R)^2\xi'(R).
\end{equation}
The requirement $\det F>0$ forces the radial deformation $\xi$ to be a monotonically increasing function; in particular $\xi$ is invertible. For homogeneous materials, the mass density in physical space becomes
\[
\rho(\psi(X))=\frac{\mathcal{K}}{n(R)}.
\]
Hence, denoting $x=\psi(X)$ and $r=|x|$, we have
\begin{equation}\label{rhor}
\rho(r)=\frac{\mathcal{K}}{n(\xi^{-1}(r))}.
\end{equation}
The most general spherically symmetric Piola-Kirchhoff  stress tensor $\Sigma$ has the form
\[
\Sigma_{IJ}(X)=-P_\mathrm{rad}(R)\frac{X_IX_J}{R^2}-P_\mathrm{tan}(R)\left(\delta_{IJ}-\frac{X_IX_J}{R^2}\right),
\]
for some functions $P_\mathrm{rad}$, $P_\mathrm{tan}$. In particular, spherically symmetric bodies are necessarily isotropic and frame indifferent. 
For homogeneous elastic bodies, the principal pressures $P_\mathrm{rad}$ and $P_\mathrm{tan}$ must be functions of the deformation gradient, and therefore can be written as
\[
P_\mathrm{rad}(R)=\widehat{P}_\mathrm{rad}(n(R),\tau(R)),\quad P_\mathrm{tan}(R)=\widehat{P}_\mathrm{tan}(n(R),\tau(R)).
\]
Substituting in~\eqref{piola} we obtain the stress tensor at the point $\psi(X)$ in physical space: 
\begin{equation}\label{cazzarola}
\sigma_{IJ}(\psi(X))=-\frac{\widehat{P}_\mathrm{rad}(n(R),\tau(R))}{\tau(R)^2}\frac{X_IX_J}{R^2}-\frac{\tau(R)}{n(R)}\widehat{P}_\mathrm{tan}(n(R),\tau(R))\left(\delta_{IJ}-\frac{X_IX_J}{R^2}\right).
\end{equation}
Using $\psi^{-1}(x)/\xi^{-1}(r)=x/r$ we obtain
\[
\sigma_{ij}(x)=-p_\mathrm{rad}(r)\frac{x_ix_j}{r^2}-p_\mathrm{tan}(r)\left(\delta_{ij}-\frac{x_ix_j}{r^2}\right),
\]
where 
\begin{align*}
&p_\mathrm{rad}(r)=\tau(\xi^{-1}(r))^{-2}\widehat{P}_\mathrm{rad}(n(\xi^{-1}(r)),\tau(\xi^{-1}(r))),\\
&p_\mathrm{tan}(r)=\frac{\tau(\xi^{-1}(r)}{n(\xi^{-1}(r))}\widehat{P}_\mathrm{tan}(n(\xi^{-1}(r)),\tau(\xi^{-1}(r))). 
\end{align*}
As $n(\xi^{-1}(r))=\mathcal{K}/\rho(r)$ and $\tau(\xi^{-1}(r))=r/\xi^{-1}(r)$, we find that
\[
p_\mathrm{rad}(r)=\widehat{p}_\mathrm{rad}(\delta(r),\eta(r)),\quad p_\mathrm{tan}(r)=\widehat{p}_\mathrm{tan}(\delta(r),\eta(r)),
\]
where 
\begin{equation}\label{deltaeta}
\delta(r)=\frac{\rho(r)}{\mathcal{K}},\quad \eta(r)=\left(\frac{\xi^{-1}(r)}{r}\right)^3,
\end{equation}
and
\begin{equation}\label{cazzacina}
\widehat{p}_\mathrm{rad}(\delta,\eta)=\eta^{2/3}\widehat{P}_\mathrm{rad}(\delta^{-1},\eta^{-1/3}),\quad \widehat{p}_\mathrm{tan}(\delta,\eta)=\frac{\delta}{\eta^{1/3}}\widehat{P}_\mathrm{tan}(\delta^{-1},\eta^{-1/3}).
\end{equation}
Moreover
\[
\eta'(r)=-3\frac{\eta(r)}{r}+3\eta(r)^{2/3}\frac{(\xi^{-1})'(r)}{r}=-3\frac{\eta(r)}{r}+3\eta(r)^{2/3}(r\xi'(\xi^{-1}(r))^{-1}.
\]
From~\eqref{nR} we have $\xi'(\xi^{-1}(r))=\eta(r)^{2/3}n(\xi^{-1}(r))$, hence, by~\eqref{rhor},
\begin{equation}\label{eqeta}
(r^3\eta(r))'=3r^2\frac{\rho(r)}{\mathcal{K}},
\end{equation}
When the body in physical space is a ball with radius $r_1$, then~\eqref{eqeta} gives
\begin{equation}\label{etaball}
\eta(r)=\frac{3\int_0^r\rho(s)s^2\,ds}{\mathcal{K}r^3}=\frac{m(r)}{\frac{4\pi}{3}\mathcal{K}r^3}\quad r\in (0,r_1),
\end{equation}
where $m(r)$ is the local mass of the ball in physical space. 
When the body in physical space is a shell supported in the interval $[r_0,r_1]$, $r_0>0$, then
\[
\eta(r)=\eta(r_0)\left(\frac{r_0}{r}\right)^3+\frac{3\int_{r_0}^r\rho(s)s^2\,ds}{\mathcal{K}r^3}=\eta(r_0)\left(\frac{r_0}{r}\right)^3+\frac{m(r)}{\frac{4\pi}{3}\mathcal{K}r^3}\quad r\in (r_0,r_1),
\]
where $m(r)$ is the local mass of the shell in physical space. Having assumed that a shell in physical space can only form when the material manifold itself is a shell, then $r_0=\xi(\mathcal{S})$, hence
\begin{equation}\label{etashell}
\eta(r_0)=\left(\frac{\xi^{-1}(r_0)}{r_0}\right)^3=\left(\frac{\mathcal{S}}{r_0}\right)^3\Rightarrow\eta(r)=\left(\frac{\mathcal{S}}{r}\right)^3+\frac{m(r)}{\frac{4\pi}{3}\mathcal{K}r^3},\quad r\in (r_0,r_1).
\end{equation}
{\it Remark.} Note that we do not allow $\mathcal{S}=0$, i.e., we disregard the possibility that a shell in physical space is obtained by deforming a ball in the reference state. This choice is motivated both on physical grounds, since gravity cannot create a cavity in a ball, and on mathematical grounds, since, as mentioned above, a deformation producing a shell from a ball would necessary be discontinuous at the center of the ball.

\subsubsection*{Example: Seth model} 
The Almansi strain tensor $A=\frac{1}{2}(\mathbb{I}-C^{-T})$ in spherical symmetry becomes
\begin{align*}
A_{IJ}&=\frac{1}{2}(\delta_{IJ}-(F^{-1})^2_{IJ})\\
&=\frac{1}{2}\left(1-\frac{\tau(R)^4}{n(R)^2}\right)\frac{X_IX_J}{R^2}+\frac{1}{2}\left(1-\frac{1}{\tau(R)^2}\right)\left(\delta_{IJ}-\frac{X_IX_J}{R^2}\right).
\end{align*}
Hence
\[
i_1(A)=\frac{1}{2}\left(3-\frac{\tau(R)^4}{n(R)^2}-\frac{2}{\tau(R)^2}\right).
\]
Replacing in~\eqref{cauchyseth} we have
\begin{align*}
\sigma(\psi(X))&=\left(\frac{3\lambda+2\mu}{2}-\frac{\lambda+2\mu}{2}\frac{\tau(R)^4}{n(r)^2}-\frac{\lambda}{\tau(R)^2}\right)\frac{X_IX_J}{R^2}\\
&\quad+\left(\frac{3\lambda+2\mu}{2}-\frac{\lambda+\mu}{\tau(R)^2}-\frac{\lambda}{2}\frac{\tau(R)^4}{n(R)^2}\right)\left(\delta_{IJ}-\frac{X_IX_J}{R^2}\right).
\end{align*}
Comparing with~\eqref{cazzarola} we see that
\begin{align*}
&\widehat{P}_\mathrm{rad}(n,\tau)=\lambda+\frac{\lambda+2\mu}{2}\frac{\tau^6}{n^2}-\frac{3\lambda+2\mu}{2}\tau^2,\\
&\widehat{P}_\mathrm{tan}(n,\tau)=\frac{\lambda}{2}\frac{\tau^3}{n}+(\lambda+\mu)\frac{n}{\tau^3}-\frac{3\lambda+2\mu}{2}\frac{n}{\tau}.
\end{align*}
Using the preceding equations in~\eqref{cazzacina} leads to the equations of state~\eqref{sethmodel} in physical space.
\subsection{Hyperelastic materials}
In this section we focus on the important case of homogeneous hyperelastic materials, which means that the the first Piola-Kirchhoff stress tensor can be written in the form
\[
\Sigma(X)=\partial_F W(F(X)),\quad \text{that is }\ \widehat{\Sigma}(F)=\partial_F W(F)
\]
for some function $W(F)$, called the stored energy function. As in spherical symmetry any material is frame indifferent and isotropic, we can restrict to a stored energy function of the form 
\[
W(F)=\Phi(\lambda_1,\lambda_2,\lambda_3)
\]
where $\lambda_1,\lambda_2,\lambda_3>0$ are the principal stretches, i.e., $\lambda_1^2,\lambda_2^2,\lambda_3^2$ are the eigenvalues of the Cauchy-Green tensor $C=F^T\cdot F$. Moreover in spherical symmetry we have 
\[
C_{IJ}(X)=\frac{n(R)^2}{\tau(R)^4}\frac{X_IX_J}{R^2}+\tau(R)^2\left(\delta_{IJ}-\frac{X_IX_J}{R^2}\right),
\]
hence the principal stretches are 
\[
\lambda_1=\frac{n(R)}{\tau(R)^2},\quad \lambda_2=\lambda_3=\tau(R)
\]
and therefore we can restrict to stored energy functions of the form
\begin{equation}\label{formW}
W(F)=\Phi(\lambda_1,\lambda_2,\lambda_2)=\phi(\lambda_1,\lambda_2).
\end{equation}
Note that $\lambda_1,\lambda_2=\lambda_3$ are the eigenvalues of $F$ and we can rewrite~\eqref{formF} as
\begin{equation}\label{formF2}
F_{IJ}(X)=\lambda_1(R)\frac{X_IX_J}{R^2}+\lambda_2(R)\left(\delta_{IJ}-\frac{X_IX_J}{R^2}\right).
\end{equation}
\begin{lemma}
The first Piola-Kirchhoff stress tensor of homogeneous hyperelastic materials in spherical symmetry is given by
\[
\Sigma(X)=-P_\mathrm{rad}(R)\frac{X_IX_J}{R^2}-P_\mathrm{tan}(R)\left(\delta_{IJ-}\frac{X_IX_J}{R^2}\right),
\]
where 
\[
P_\mathrm{rad}(R)=-\partial_1\phi(\lambda_1(R),\lambda_2(R)),\quad P_\mathrm{tan}(R)=-\frac{1}{2}\partial_2\phi(\lambda_1(R),\lambda_2(R)).
\]
\end{lemma}
\begin{proof}
We have
\begin{equation}\label{temporal}
\Sigma=(\partial_1\phi)\frac{\partial \lambda_1}{\partial F}+(\partial_2\phi)\frac{\partial \lambda_2}{\partial F}.
\end{equation}
Using that $\partial_A\det A=(\det A) A^{-1}$ and $\partial_A\mathrm{Tr}A=\mathbb{I}$ we obtain
\[
\frac{1}{\lambda_1}\frac{\partial \lambda_1}{\partial F}+\frac{2}{\lambda_2}\frac{\partial \lambda_2}{\partial F}=F^{-1},\quad
\frac{\partial \lambda_1}{\partial F}+2\frac{\partial \lambda_2}{\partial F}=\mathbb{I}.
\]
It follows that
\[
\frac{\partial \lambda_1}{\partial F}=-\frac{\lambda_1}{\lambda_2-\lambda_1}(\mathbb{I}-\lambda_2 F^{-1}),\quad \frac{\partial \lambda_1}{\partial F}=\frac{\lambda_2}{2(\lambda_2-\lambda_1)}(\mathbb{I}-\lambda_1 F^{-1}).
\]
As
\[
(F^{-1})_{IJ}=\frac{1}{\lambda_1}\frac{X_IX_J}{R^2}+\frac{1}{\lambda_2}\left(\delta_{IJ}-\frac{X_IX_J}{R^2}\right),
\]
we further have
\[
\frac{\partial \lambda_1}{\partial F_{IJ}}=\frac{X_IX_J}{R^2},\quad
\frac{\partial \lambda_2}{\partial F_{IJ}}=\frac{1}{2}\left(\delta_{IJ}-\frac{X_IX_J}{R^2}\right).
\]
Replacing in~\eqref{temporal} yields the result.
\end{proof}
As in spherical symmetry $F=F^T$ is given by~\eqref{formF2} and $\det F=\lambda_1\lambda_2^2$,~\eqref{piola} gives
%
\[
\sigma_{IJ}(\psi(X))=\frac{\partial_1\phi(\lambda_1(R),\lambda_2(R))}{\lambda_2(R)^2}\frac{X_IX_J}{R^2}+\frac{1}{2}\frac{\partial_2\phi(\lambda_1(R),\lambda_2(R))}{\lambda_1(R)\lambda_2(R)}\left(\delta_{IJ}-\frac{X_IX_J}{R^2}\right).
\]
Hence, using that $\psi^{-1}(x)/\xi^{-1}(r)=x/r$, we obtain
\[
\sigma_{ij}(x)=-p_\mathrm{rad}(r)\frac{x_ix_j}{r^2}-p_\mathrm{tan}(r)\left(\delta_{ij}-\frac{x_ix_j}{r^2}\right),
\]
where
\[
p_\mathrm{rad}(r)=-\frac{\partial_1\phi(\lambda_1(\xi^{-1}(r)),\lambda_2(\xi^{-1}(r))),}{\lambda_2(\xi^{-1}(r))^2},\quad
p_\mathrm{tan}(r)=-\frac{\partial_2\phi(\lambda_1(\xi^{-1}(r)),\lambda_2(\xi^{-1}(r)))}{2\lambda_1(\xi^{-1}(r))\lambda_2(\xi^{-1}(r))}.
\]
Defining $\delta(r)=\rho(r)/\mathcal{K}$, $\eta(r)=1/\tau(\xi^{-1}(r))^3=1/\lambda_2(\xi^{-1}(r))^3$ as in~\eqref{deltaeta}
and using that, by~\eqref{rhor},
\[
\lambda_1(\xi^{-1}(r))=\frac{\eta(r)^{2/3}}{\delta(r)},
\]
we obtain
\[
p_\mathrm{rad}(r)=\widehat{p}_\mathrm{rad}(\delta(r),\eta(r)),\quad \widehat{p}_\mathrm{tan}(r)=p_\mathrm{tan}(\delta(r),\eta(r)).
\]
where 
\[
\widehat{p}_\mathrm{rad}(\delta,\eta)=-\eta^{2/3}\partial_1\phi(\delta^{-1}\eta^{2/3},\eta^{-1/3}),\quad \widehat{p}_\mathrm{tan}(\delta,\eta)=-\frac{1}{2}\frac{\delta}{\eta^{1/3}}\partial_2\phi(\delta^{-1}\eta^{2/3},\eta^{-1/3}).
\]
Finally, letting 
\[
w(\delta,\eta)=\phi(\delta^{-1}\eta^{2/3},\eta^{-1/3})
\]
we obtain
\[
\widehat{p}_\mathrm{rad}(\delta,\eta)=\delta^2\partial_\delta w(\delta,\eta),\quad \widehat{p}_\mathrm{tan}(\delta,\eta)=\widehat{p}_\mathrm{rad}(\delta,\eta)+\frac{3}{2}\delta\eta\partial_\eta w(\delta,\eta).
\]
Moreover due to~\eqref{formW} and the compatibility conditions~\eqref{lincom} with linear elasticity, we have
\[
\partial_1^2\phi(1,1)=\lambda+2\mu,\quad\partial_2^2\phi(1,1)=4\lambda+4\mu,\quad \partial_1\partial_2\phi(1,1)=2\lambda.
\]
For materials satisfying the natural state condition~\eqref{naturalstatecon} the previous equations 
become~\eqref{com1}-\eqref{com2}.
\subsubsection*{Examples}

For the quasi-linear Signorini model in spherical symmetry we have
\begin{align*}
W(F)=\phi(\lambda_1,\lambda_2)=\frac{\lambda_1\lambda_2^2}{8}\left[4\mu\left(\frac{1}{\lambda_1^2}+\frac{2}{\lambda_2^2}-1\right)+(\lambda+\mu)\left(\frac{1}{\lambda_1^2}+\frac{2}{\lambda_2^2}-2\right)^2\right]-\mu.
\end{align*}
Hence $w(\delta,\eta) =\phi(\delta^{-1}\eta^{2/3},\eta^{-1/3})$ is given by
\begin{align*}
w(\delta,\eta) =\frac{1}{\delta}\left[\frac{1}{8} \left(\frac{\delta ^2}{\eta ^{4/3}}+2 \eta ^{2/3}-3\right)^2 (\lambda
   +\mu )+\frac{\mu}{2} \left(\frac{\delta ^2}{\eta ^{4/3}}+2 \eta
   ^{2/3}-1\right)\right]-\mu.
\end{align*}
Similarly one obtains the other stored energy functions presented in Section~\ref{examples}.
\section{Appendix: Variational formulation for hyperelastic materials}\label{varsec}
In this appendix we show that the equation~\eqref{generalss} for static, spherically symmetric,  self-gravitating hyperelastic bodies admits a variational formulation. It is convenient to define
\[
\Psi(\rho,\eta)=\mathcal{K}^{-1}w(\rho/\mathcal{K},\eta),
\] 
so that~\eqref{hyperdef} becomes
\begin{equation}\label{hyperdef2}
p_\mathrm{rad}(r)=\rho(r)^2\partial_\rho\Psi(\rho(r),\eta(r)),\quad p_\mathrm{tan}(r)=p_\mathrm{rad}(r)+\frac{3}{2}\rho(r)\eta(r)\partial_\eta\Psi(\rho(r),\eta(r)).
\end{equation}
Moreover we assume for simplicity that $r\in (0,\infty)$; a similar argument applies to the case $r\in (r_0,\infty)$ with $r_0>0$. Hence the function $\eta$ is given by
\[
\eta(r)=\frac{m(r)}{\frac{4\pi}{3}\mathcal{K}r^3},\quad m(r)=4\pi\int_0^r\rho(s)s^2\,ds,\quad r>0.
\]
Consider the functional
\[
E[\rho]=\int_0^\infty \rho(r) \Psi(\rho(r),\eta(r)) r^2\,dr-\frac{1}{8\pi}\int_0^\infty \left(\frac{m(r)}{r^2}\right)^2 r^2\,dr.
\]
If $\rho_0$ is a minimizer and we set $\rho_\tau(r)=\rho_0(r)+\tau\phi(r)$, where $\phi$ is a smooth function with compact support, we obtain
\begin{align*}
0=\left(\frac{d}{d\tau}E[\rho_\tau]\right)_{\tau=0}&=\int_0^\infty\partial_\rho(\rho \Psi)(\rho_0(r),\eta_0(r))\phi(r) r^2\,dr\\
&\quad+\int_0^\infty \frac{3\rho_0(r)}{\mathcal{K}}\partial_\eta \Psi(\rho_0(r),\eta_0(r))\left(\frac{1}{r^2}\int_0^r s^2\phi(s)\,ds\right)r\,dr\\
&\quad-\int_0^\infty\frac{m_0(r)}{r^2}\left(\int_0^r s^2\phi(s)\,ds\right)dr.
\end{align*}
Exchanging the order of integration in the last integral we find
\[
\int_0^\infty\frac{m_0(r)}{r^2}\left(\int_0^r s^2\phi(s)\,ds\right)\,dr=-\int_0^\infty  V_0(s)\phi (s) s^2 ds,
\]
where $V_0(s)=-\int_s^\infty m_0(r)/r^2\,dr$ is the potential induced by the minimizer. Similarly the second integral is
\[
\int_0^\infty \frac{3\rho_0(r)}{\mathcal{K}}\partial_\eta \Psi(\rho_0(r),\eta_0(r))\left(\frac{1}{r^2}\int_0^r s^2\phi(s)\,ds\right)r^2\,dr=\int_0^\infty\Lambda_0(s)\phi(s)s^2\,ds,
\]
where
\[
\Lambda_0(s)=\int_s^\infty\frac{3\rho_0(r)}{\mathcal{K}}\partial_\eta \Psi(\rho_0(r),\eta_0(r))\,\frac{dr}{r}.
\]
Since $\phi$ is arbitrary it must hold that $\partial_\rho(\rho \Psi)(\rho_0(r),\eta_0(r))+\Lambda_0(r)+V_0(r)=0$.
Differentiating with respect to $r$ we obtain
\begin{equation}\label{agua}
\frac{d}{dr}\partial_\rho(\rho \Psi)(\rho_0(r),\eta_0(r))-\frac{3}{r}\frac{\rho_0(r)}{\mathcal{K}}\partial_\eta \Psi(\rho_0(r),\eta_0(r))+\frac{d}{dr}V_0(r)=0.
\end{equation}
Letting $p_\mathrm{rad}^0(r)=\rho_0(r)^2\partial_\rho \Psi(\rho_0(r),\eta_0(r))$ be the radial pressure of the minimizer and using that 
\[
\eta'(r)=\frac{3}{r}\big(\frac{\rho(r)}{\mathcal{K}}-\eta(r)\big),
\]
we have
\begin{align*}
\frac{d}{dr}[\partial_\rho(\rho \Psi)(\rho_0(r),\eta_0(r))]&=-3\frac{\eta_0(r)}{r}\partial_\eta \Psi(\rho_0(r),\eta_0(r))+\frac{3}{r}\frac{\rho_0(r)}{\mathcal{K}}\partial_\eta \Psi(\rho_0(r),\eta_0(r))\\
&\quad+\frac{1}{\rho_0(r)}\frac{d}{dr}p_\mathrm{rad}^0(r).
\end{align*}

Replacing in~\eqref{agua} and using $\frac{d}{dr}V_0(r)=m_0(r)/r^2$, as well as $p_\mathrm{rad}-p_\mathrm{tan}=-\frac{3}{2}\rho\eta\partial_\eta \Psi$, we obtain
\[
\frac{1}{\rho_0(r)}\frac{d}{dr}p_\mathrm{rad}^0(r)+\frac{2}{r}\frac{p_\mathrm{rad}^0(r)-p_\mathrm{tan}^0(r)}{\rho_0(r)}+\frac{m_0(r)}{r^2}=0,
\]
which is~\eqref{generalss} inside the support of the matter. 

\end{appendices}

\end{document}